\documentclass[11pt,a4paper]{article}

\usepackage[comma]{natbib}
\usepackage{anysize}
\usepackage{color}
\usepackage{graphicx}
\usepackage[utf8]{inputenc}
\usepackage{pgfplots}
\pgfplotsset{compat=newest}
\usepackage{ulem}

\usepackage{amsmath,amssymb,amsthm,amsopn}
\usepackage[ruled]{algorithm2e}
\renewcommand{\vec}[1]{\boldsymbol{#1}}
\newcommand{\gvec}[1]{\boldsymbol{#1}}
\newcommand{\mat}[1]{\mathbf{#1}}

\newtheorem{lemma}{Lemma}
\newtheorem{theorem}{Theorem}
\theoremstyle{remark}

\DeclareMathOperator{\E}{\boldsymbol{\mathbb{E}}}

 \DeclareMathOperator{\tr}{tr}

\def\SCV{\mathrm{SCV}}

\def\PI{\mathrm{PI}}
\def\CV{\mathrm{CV}}

\def\comm{\mat{K}}

\def\Herm{\boldsymbol{\mathcal{H}}}

\renewcommand{\today}{\begingroup
\number \day\space  \ifcase \month \or January\or February\or
March\or April\or May\or June\or July\or August\or September\or
October\or November\or December\fi \space  \number \year \endgroup}

\marginsize{3cm}{3cm}{3cm}{3cm}

\newcommand{\D}{\mathsf{D}}
\newcommand{\bH}{{\bf H}}
\newcommand{\bG}{{\bf G}}
\newcommand{\bI}{{\bf I}}

\newcommand{\bx}{{\boldsymbol x}}

\newcommand{\bz}{{\boldsymbol z}}
\newcommand{\bX}{{\bf X}}

\newcommand{\be}{{\boldsymbol e}}
\newcommand{\bmu}{{\boldsymbol \mu}}
\newcommand{\bpsi}{{\boldsymbol \psi}}

\renewcommand{\S}{\boldsymbol{\mathcal S}}
\newcommand{\bSigma}{{\boldsymbol \Sigma}}

\renewcommand{\vec}{\operatorname{vec}}

\theoremstyle{plain}

\newtheorem{cor}{Corollary}
\newtheorem{teor*}{Teorema}

\theoremstyle{definition}

\title{Efficient recursive algorithms for functionals based on higher order derivatives of the multivariate Gaussian density}

\author{José E. Chacón\footnote{Departamento de
Matemáticas, Universidad de Extremadura, E-06006 Badajoz, Spain. E-mail:
{\tt jechacon@unex.es}}\ \ and Tarn Duong\footnote{{Sorbonne Universities,
University Pierre and Marie Curie (UPMC) -- Paris 6,
Theoretical and Applied Statistics Laboratory (LSTA), UR 1, F-75005, Paris, France.}}
\footnote{{Assistance Publique-Hôpitaux de Paris (AP-HP), Pitié-Salpêtrière Hospital, Institute of Translational Neurosciences (IHU-A-ICM), F-75005 Paris, France.
Email: {\tt tarn.duong@upmc.fr}}}}
\date{\today}

\begin{document}

\maketitle

\begin{abstract}
\noindent Many developments in Mathematics involve the computation of higher order
derivatives of Gaussian density functions. The analysis of univariate Gaussian random variables is
a well-established field whereas the analysis of their multivariate counterparts consists of a body
of results which are more dispersed. These latter results generally fall into two main categories:
theoretical expressions which reveal the deep structure of the problem, or computational algorithms
which can mask the connections with closely related problems. In this paper, we unify existing
results and develop new results in a framework which is both conceptually cogent and
computationally efficient. We focus on the underlying connections between higher order derivatives
of Gaussian density functions, the expected value of products of quadratic forms in Gaussian
random variables, and $V$-statistics of degree two based on Gaussian density functions. These three
sets of results are combined into an analysis of non-parametric data smoothers.

\end{abstract}



\section{Introduction}

Gaussian random variables and their associated probability density functions are commonly studied
in Statistics since they possess many attractive theoretical and computational properties. In fact,
Gaussian functions and its derivatives appear as fundamental tools in many areas of Mathematics,
and also in other disciplines like Physics or Engineering.

Many results have been established for univariate Gaussian random variables in a unified framework.
For multivariate random variables, many results are available as well, but due to the lack of a
commonly accepted notation for higher order derivatives of the multivariate functions, these are
more scattered.

In this paper we adopt the vectorized form of higher order multidimensional derivatives, which was
the key tool that allowed to obtain explicit formulas for the moments of the multivariate normal
distribution of arbitrary order \citep{Hol88} and for the higher order derivatives of the
multivariate Gaussian density function, through the introduction of vector Hermite polynomials
\citep{Hol96a}.

These polynomials, however, depend on a matrix (so-called symmetrizer matrix) having an enormous
number of entries, even when the dimension and the derivative order are not high. Thus, although
these results provide a general formulation that is valid and useful for developing theory in any
dimension and for an arbitrary derivative/moment order, some authors like \cite{Tri03} or
\cite{Kan08} pointed out the difficulties that the computation of such a large matrix represents in
practical situations.

Here we unify existing results, as well as developing new ones, in a cogent framework which
facilitates a concise theoretical form as well as an efficient computational form.
We begin by
introducing the vectorized form of the higher order derivatives of the multivariate Gaussian
density functions, via their factorization as Hermite polynomials, in Section~\ref{sec:deriv}.
Efficient recursive algorithms to compute the involved high-dimensional symmetrizer matrix, and the product of this matrix and a high-dimensional vector, are discussed in Sections \ref{sec:symm} and \ref{sec:4}, respectively. A different approach is developed in Section \ref{sec:unique}, by focusing on the recursive computation of the unique partial derivative operators that unequivocally determine the full derivative vector. We then focus on some statistical applications intimately linked with the derivatives of the multivariate Gaussian density function: the computation of moments of multivariate normal distributions and the expectation of powers of quadratic forms in Gaussian random variables are explored in Sections \ref{sec:moment} and \ref{sec:quad}, respectively, including a new result providing a formula for the joint cumulants of quadratic forms in normal variables that corrects an identity included in \cite{MP92} that is not correct in general. In Section~\ref{sec:kde} we show how these functionals are extremely useful for the analysis of non-parametric data smoothers, which involve the computation of $V$-statistics of degree two based on derivatives of multivariate Gaussian density functions. Finally, in Section \ref{sec:comparisons} all the newly introduced recursive algorithms are compared to the standard, direct approach, in terms of computation time.

\color{black}

\section{Higher order derivatives of Gaussian density functions}
\label{sec:deriv}

The characterization of the $r$-th order derivatives of a $d$-variate function can be expressed in
many ways using, e.g. matrices, tensors or iterated permutations. We use the characterization using
Kronecker products of vectors, popularized by \citet{Hol96a}. Let $f$ be a real $d$-variate
function, $\bx=(x_1,\dots,x_d)$, and $\D=\partial/\partial\bx=(\partial/\partial
x_1,\dots,\partial/\partial x_d)$ be the first derivative (gradient) operator. If the usual
convention $(\partial/\partial x_i)(\partial/\partial x_j) = \partial^2/( \partial x_i\partial
x_j)$ is taken into account, then the $r$-th derivative of $f$ is defined to be the vector
$\D^{\otimes r} f= (\D f)^{\otimes r}=\partial^r f/\partial{\bx}^{\otimes r}\in\mathbb R^{d^r}$,
with $\D^{\otimes 0} f=f$, $\D^{\otimes 1}f=\D f$. Here, $\mathsf D^{\otimes r}$ refers to the
$r$-th Kronecker power of the operator $\mathsf D$, formally understood as the $r$-fold product
$\mathsf D\otimes\cdots\otimes\mathsf D$.

For example,  all the second order partial derivatives can be organized into the usual Hessian
matrix $\mathsf Hf=(\partial^2f/\partial x_i\partial x_j)_{i,j=1}^d$, and the Hessian operator can
be formally written as $\mathsf H=\D\D^\top$. The equivalent vectorized form is $\D^{\otimes
2}=\vec \mathsf H$, where $\vec$ denotes the operator which concatenates the columns of a matrix
into a single vector, see \citet{HS79}.

For Hessian matrices, there is not much gain from using this vectorized form since the matrix form
is already widely analyzed. However for $r>2$, this vectorized characterization, which maintains
all derivatives as vectors, has contributed to recent advances in multivariate analysis which have
been long hindered by the lack of suitable analytical tools. For example, \citet{CD10,CD11} and
\citet{CDW11} treated higher derivatives involved in multivariate non-parametric data smoothing.
These authors relied heavily on the derivatives of the Gaussian density function, as defined in
terms of the Hermite polynomials.

Let $\phi(\bx)=(2\pi)^{-d/2}\exp(-\frac12\bx^\top\bx)$ be the standard $d$-dimensional Gaussian density and $\phi_\bSigma(\bx)=|\bSigma|^{-1/2}\phi(\bSigma^{-1/2}\bx)$ be the centred Gaussian density with variance $\bSigma$. \citet{Hol96a} showed that the $r$-th derivative of $\phi_\bSigma$ is
\begin{equation} \label{Drphi}
\D^{\otimes r} \phi_\bSigma( \bx) = (-1)^r (\bSigma^{-1})^{\otimes r} \Herm_r(\bx; \bSigma) \phi_\bSigma( \bx),
\end{equation}
\color{black}
where the $r$-th order Hermite polynomial is defined by
\begin{equation}\label{Herm}
 \Herm_r(\bx; \bSigma)  = r! \S_{d,r} \sum_{j=0}^{[r/2]}  \frac{(-1) ^{j}}{j! (r-2j)! 2^j }
\big\{\bx^{\otimes (r - 2j)} \otimes  (\vec \bSigma)^{\otimes j} \big\}.
\end{equation}
Here $\S_{d,r}$ is the $d^r \times d^r$ symmetrizer matrix defined as
\begin{equation}\label{expdef}
\S_{d,r}=\frac{1}{r!}\sum_{i_1,i_2,\dots,i_r=1}^d\;\sum_{\sigma\in\mathcal
P_r}\bigotimes_{\ell=1}^r\be_{i_\ell}\be_{i_{\sigma(\ell)}}^T,
\end{equation}
with $\mathcal P_r$ standing for the group of permutations of order $r$ and $\be_i$ for the $i$th
column of $\bI_d$, the identity matrix of order $d$. We also have that $\S_{d,0}=1$ and $\S_{d,1}=\bI_d$.
This definition is highly abstract so we take a concrete example to demonstrate
the action of this symmetrizer matrix on a 3-fold product, i.e.
$\boldsymbol{\mathcal{S}}_{d,3}
(\bx_1 \otimes \bx_2 \otimes \bx_3) = \frac{1}{6}
[ \bx_1 \otimes \bx_2 \otimes \bx_3
+ \bx_1 \otimes \bx_3 \otimes \bx_2
+ \bx_2 \otimes \bx_1 \otimes \bx_3
+ \bx_2 \otimes \bx_3 \otimes \bx_1
+ \bx_3 \otimes \bx_1 \otimes \bx_2
+ \bx_3 \otimes \bx_2 \otimes \bx_1]$.
In general, the symmetrizer matrix $\S_{d,r}$ maps the product
$\bigotimes_{i=1}^r \bx_i$ to
an equally weighted linear combination of products of
all possible permutations of $\bx_1, \dots, \bx_r$.

The goal of this paper is to investigate efficient ways to compute the $r$-th derivative
$\D^{\otimes r} \phi_\bSigma( \bx)$ of the multivariate Gaussian density function, and their
applications to several statistical problems.

\section{Recursive computation of the symmetrizer matrix}
\label{sec:symm}

Surely the most prohibitive element in the computation of the $r$-th derivative of the $d$-variate Gaussian density is the symmetrizer matrix $\S_{d,r}$. It is a huge matrix, even for low values of $d$ and $r$ (for instance, $\S_{4,8}$ is a matrix of order $65536\times 65536$) and its definition involves $r!d^r$ summands, hence its direct calculation can be onerous both in memory storage and computational time.

Nevertheless, this symmetrizer matrix has independent interest on its own from an algebraic point of view. This is well certified by the fact that it has been independently discovered many times. To our knowledge, \cite{Hol85} was the first to develop its form as a generalization of Kronecker product permuting matrices. More recently, \cite{Sch03} and \cite{Mei05} found alternative derivations and further interesting properties.

First, to reduce the number of loops in (\ref{expdef}) it is useful to consider the conversion to
base $d$. Any number $i\in\{0,1,\dots,d^r-1\}$ can be written in base $d$ as $(a_r\ldots
a_2a_1)_d$, with digits $a_j\in\{0,1,\dots,d-1\}$, meaning that $i=\sum_{j=1}^ra_jd^{j-1}$. A
simple translation yields that the correspondence between the set
$\mathcal{PR}_{d,r}=\big\{(i_1,\dots,i_r)\colon i_1,\dots,i_r\in\{1,\dots,d\}\big\}$ of
permutations with repetition of $d$ elements, taken $r$ at a time, and the set $\{1,2,\dots,d^r\}$,
given by $p(i_1,\dots,i_r)=1+\sum_{j=1}^r(i_j-1)d^{j-1}$ is also bijective, hence all $r$-tuples
$(i_1,\dots,i_r)$ involved in the multi-index for the first summation in (\ref{expdef}) can be
obtained as $p^{-1}(i)$ as $i$ ranges over $\{1,2,\dots,d^r\}$ (see Appendix \ref{pinv}), so that
only two loops are needed for the direct computation of the symmetrizer matrix. Moreover, after a
careful inspection of the $r$-fold Kronecker product involved it follows that (\ref{expdef}) can be
written as
\begin{equation}\label{expdef2}
\S_{d,r}=\frac{1}{r!}\sum_{i=1}^{d^r}\sum_{\sigma\in\mathcal
P_r}\mat E_{i,(p\circ\sigma\circ p^{-1})(i)},
\end{equation}
where $\mat E_{i,j}$ represents the $d^r\times d^r$ matrix having the $(i, j)$-th element equal to
1 as its only nonzero element. The operator $p^{-1}$ maps an integer $i$ to a unique $r$-tuple
$(i_1, \dots, i_r)$, the operator $\sigma$ generates a permutation of a given $r$-tuple, and the
operator $p$ maps an $r$-tuple to an integer in $\lbrace 1, \dots, d^r \rbrace$. So the composition
$(p\circ\sigma\circ p^{-1})$, as $i$ ranges over $\lbrace 1, \dots, d^r \rbrace$ and $\sigma$ over
$\mathcal P_r$, generates an equivalent set to the set of permutations defined in~(\ref{expdef}).
Hence, the novel formulation in Equation~(\ref{expdef2}) is more appropriate for efficient
computations.

Even in this simple form, the direct implementation of $\S_{d,r}$ using (\ref{expdef2}) usually
takes a considerable amount of time as $d$ and $r$ increase, due to the large number of terms
involved in each of the two loops. A useful way to improve over the direct approach is to use a
recursive implementation of $\S_{d,r}$. Thus, the goal of this section is to express the
symmetrizer matrix $\S_{d,r+1}$ in terms of the symmetrizer matrix of lower order $\S_{d,r}$.

Let us denote by $\mathcal M_{m\times n}$ the set of all $m\times n$ matrices and let $\comm_{r,s} \in \mathcal{M}_{rs \times rs}$ be the commutation matrix of order $r,s$; see \citet*{MN79}. The commutation matrix allows us to commute the order of the matrices in a Kronecker
product, e.g., if $\mat A \in  \mathcal{M}_{m \times n}$ and $\mat B \in  \mathcal{M}_{p \times
q}$, then $\comm_{p,m} (\mat A \otimes \mat B) \comm_{n,q} = \mat B \otimes \mat A$. The relationship between $\S_{d,r+1}$ and
$\S_{d,r}$ is given in the next theorem.

\begin{theorem}\label{thm:1}
Consider the matrix $\mat T_{d,r}\in\mathcal M_{d^r\times d^r}$ defined by
$$\displaystyle \mat T_{d,r}=\frac1r\sum_{j=1}^r(\bI_{d^j}\otimes\mat K_{d^{r-j-1},d})(\bI_{d^{j-1}}\otimes\mat K_{d,d^{r-j}})$$
where, by convention, $\mat K_{d^{-1},d}=1\in\mathbb R$. Then $\S_{d,r+1}=(\S_{d,r}\otimes\bI_d)
\mat T_{d,r+1}$.
\end{theorem}

From Theorem \ref{thm:1} it follows that, to obtain a recursive formula for $\S_{d,r}$, it suffices to obtain
a recursive formula for $\mat T_{d,r}$. This is provided in the next result.

\begin{theorem}\label{thm:2}
For any $r\geq1$ the relationship between $\mat T_{d,r+1}$ and $\mat T_{d,r}$ is given by
$$(r+1)\mat T_{d,r+1}=(\bI_{d^{r-1}}\otimes\mat K_{d,d})(r\mat
T_{d,r}\otimes\bI_d)(\bI_{d^{r-1}}\otimes\mat K_{d,d})+\bI_{d^{r-1}}\otimes\mat K_{d,d}.$$
\end{theorem}
\enlargethispage{.1cm}
Combining Theorems \ref{thm:1} and \ref{thm:2}, the proposed  recursive algorithm to compute $\S_{d,r}$ reads as follows:

\begin{algorithm}[h!t]
\caption{Recursive symmetrizer matrix computation}\label{RecSym}
\SetKwInOut{Input}{Input}\SetKwInOut{Output}{Output}
\Input{dimension $d$ and order $r$}
\Output{Symmetrizer matrix $\S_{d,r}$}
\begin{enumerate}
\item If $r=0$ set $\S_{d,r}=1$
\item If $r=1$ set $\S_{d,r}=\bI_d$
\item If $r\geq2$ then set $\mat S=\mat T=\mat I_d$ and $\mat A=\mat K_{d,d}$

For $i$ in $2,\dots,r$:

\qquad Set $\mat T=\mat A(\mat T\otimes \bI_d)\mat A+\mat A$ and $\mat S=(\mat S\otimes \bI_d)\mat T$

\qquad If $i<r$, set $\mat A=\bI_d\otimes \mat A$

\item Return $\S_{d,r}=\mat S/r!$
\end{enumerate}
\end{algorithm}

The proofs of all the new results in the paper, including Theorems~\ref{thm:1} and \ref{thm:2},
will be deferred to Appendix \ref{proofs}. Besides, a detailed comparison of the computation times
for the direct approach (based on Equation (\ref{expdef2})) and the new recursive
Algorithm~\ref{RecSym} is given below in Section \ref{sec:comparisons}.

\section{Recursive computation of the product of the symmetrizer matrix and a vector}\label{sec:4}

Although the computation of symmetrizer matrices has an algebraic interest on its own, recall from
the Introduction that the primary motivation for the name of the symmetrizer matrix is its
symmetrizing action on a Kronecker product vector. Thus, for a vector $\boldsymbol
v=(v_1,\dots,v_{d^r})$, the product $\S_{d,r} \boldsymbol v$ deserves to be studied more closely.
For example, when the final goal is to obtain the $r$-th order Hermite polynomial it may not be
strictly necessary to compute $\S_{d,r}$ explicitly. To understand this notice that, from
(\ref{expdef2}), the $i$-th coordinate of the vector $\boldsymbol w=\S_{d,r}\boldsymbol v$ is just
\begin{equation}\label{Sdrv}
w_i=\frac{1}{r!}\sum_{\sigma\in\mathcal P_r}v_{(p\circ\sigma\circ p^{-1})(i)}.
\end{equation}
This makes it feasible to obtain $\S_{d,r}\boldsymbol v$ for higher values of $d$ and $r$, in situations where memory limitations do not allow us to compute the whole matrix $\S_{d,r}$.

The recursive approach to compute $\S_{d,r}\boldsymbol v$ is based on the following corollary of
Theorem \ref{thm:1}, in which we show that by induction it is possible to obtain a new
representation of the symmetrizer matrix, a factorization with $r$ factors depending only on the
$\mat T_{d,k}$ matrices for $k=1,\dots,r$.

\begin{cor}\label{cor:1}
For any $r=1,2,\ldots$, the symmetrizer matrix can be factorized as
$$\S_{d,r}=\prod_{k=1}^r(\mat T_{d,k}\otimes\bI_{d^{r-k}})=(\mat T_{d,1}\otimes\bI_{d^{r-1}})(\mat
T_{d,2}\otimes\bI_{d^{r-2}})\cdots (\mat T_{d,r-1}\otimes\bI_{d})
\mat T_{d,r}.$$
This factorization can be further simplified by noting that $\mat T_{d,1}=\bI_d$.
\end{cor}

Corollary \ref{cor:1} suggests a straightforward recursive scheme provided a simple formula for
each of the factors $(\mat T_{d,k}\otimes \bI_{d^{r-k}})\boldsymbol v$ is available. We derive such
a formula in the next result.

\begin{cor}\label{cor:2}
Denote by $\tau_{jk}$ the transposition that interchanges the $j$-th and $k$-th coordinates of an index vector $(i_1,\dots,i_r)$ with $1\leq i_1,\dots,i_r\leq d$. For any vector $\boldsymbol v=(v_1,\dots,v_{d^r})\in\mathbb R^{d^r}$ and $k=2,\dots, r$, it is possible to express
$(\mat T_{d,k}\otimes \bI_{d^{r-k}})\boldsymbol v=\frac1k\sum_{j=1}^k\boldsymbol w_{p\circ\tau_{jk}\circ p^{-1}}$, where $\boldsymbol w_{p\circ\tau_{jk}\circ p^{-1}}\in\mathbb R^{d^r}$ is the vector whose $(p\circ\tau_{jk}\circ p^{-1})(i)$-th coordinate is $v_i$.
\end{cor}

From Corollary \ref{cor:2} it follows that, once the set $\mathcal{PR}_{d,r}=\{p^{-1}(i)\colon
i=1,\dots,d^r\}$ of permutations with repetitions has been obtained (see Appendix \ref{pinv}), the
vector $\S_{d,r}\boldsymbol v$ can be computed using just two nested loops with a small number of
iterations, namely for $k=2,\dots,r$ and $j=1,\dots,k$. This implementation is described in
Algorithm \ref{RecSdrv}. Again, we refer to Section \ref{sec:comparisons} for the comparison of the
computation times of the direct approach (based on Equation (\ref{Sdrv})) and the new recursive
Algorithm \ref{RecSdrv}.

\begin{algorithm}[h!t]
\caption{Recursive computation of $\S_{d,r}\boldsymbol v$}\label{RecSdrv}
\SetKwInOut{Input}{Input}\SetKwInOut{Output}{Output}
\Input{dimension $d$, order $r$, a vector $\boldsymbol v\in\mathbb R^{d^r}$}
\Output{The product $\boldsymbol w=\S_{d,r}\boldsymbol v$}
\begin{enumerate}
\item Set $\boldsymbol w_{old}=\boldsymbol v$
\item If $r\geq2$ then
\begin{enumerate}
\item Generate the set $\mathcal{PR}_{d,r}$ as described in Appendix \ref{pinv}
\item For $k$ in $2,\dots,r$:

\qquad Initialize $\boldsymbol w_{new}=\boldsymbol 0\in\mathbb R^{d^r}$

\qquad For $j$ in $1,\dots,k$:

\qquad\qquad Add $\boldsymbol w_{old}$ to the coordinates of $\boldsymbol w_{new}$ reordered according to

\qquad\qquad $p\circ\tau_{jk}\circ p^{-1}$ (as indicated in Corollary \ref{cor:2})

\qquad Set $\boldsymbol w_{old}=\boldsymbol w_{new}/k$
\end{enumerate}
\item Return $\boldsymbol w_{old}$
\end{enumerate}
\end{algorithm}

\section{Recursive computation of all the unique partial derivatives of the
multivariate Gaussian density}\label{sec:unique} Employing the vectorization $\D^{\otimes r}f$ to
encompass all the $r$-th order partial derivatives into a single vector is quite useful for a neat
theoretical analysis of quantities based on multivariate higher-order derivatives. For instance,
from the explicit formula for $\D^{\otimes r}\phi_\bSigma$ given in terms of the multivariate
Hermite polynomials, involving $\S_{d,r}$, \cite{Hol88} was able to derive explicit expressions for
the moments and cumulants of arbitrary order of the multivariate normal distribution, whereas all
the previous studies only presented tailored formulas for a few particular cases (see Section
\ref{sec:moment} below).

However, many of the partial derivative operators in the vector $\D^{\otimes r}$ may appear duplicated, due to Schwarz's theorem on the commutation of higher-order partial derivatives. Thus, it would be desirable in practice to avoid computing these elements more than once. For example, when commutation of partial derivatives of second order is allowed, it suffices to compute the terms $\partial^2/\partial x_i^2$ for $i=1,\dots,d$, and just $\partial^2/(\partial x_i\partial x_j)$ for $i<j$, to obtain the whole operator $\D^{\otimes 2}$. It is not necessary to compute the mixed partial derivatives for $i>j$.
We will refer to this reduced set of partial derivatives, that unequivocally determine the full derivative vector, as the `unique partial derivatives'.
By this phrase, we mean the set of partial derivatives with unique partial derivative {\it indices}.
%\textcolor{red}{\sout{For example, for $d=2, r=2$ this is $ \lbrace \partial^2 f/\partial x_1^2, \partial^2 f/(\partial x_1 \partial x_2),
%\partial^2 f/\partial x_2^2 \rbrace$. That is, the uniqueness is in the indices, rather than in the value of
%the derivatives. For $f(x_1, x_2) = \exp(x_1+x_2)$, whilst $\partial^2 f/\partial x_1^2
%= \partial^2 f/(\partial x_1 \partial x_2) = \exp(x_1+x_2)$, for all $x_1, x_2$,
%they are considered to be distinct
%since their derivative indices are not a permutation of each other,
%and hence they are elements of this minimal set of unique partial derivatives.}}

This section makes use of this observation to introduce a different approach, in which a further
reduction in storage space and computation time is achieved by computing only the unique partial
derivatives of $\D^{\otimes r}\phi_\bSigma$ and re-distributing them later to form the full vector.

First, notice that each coordinate of the operator $\D^{\otimes r}$ can be written as
$\D_{\boldsymbol i}$ for some $\boldsymbol i\in\mathcal{PR}_{d,r}=\big\{(i_1,\dots,i_r)\colon
i_1,\dots,i_r\in\{1,\dots,d\}\big\}$, where
$$\D_{\boldsymbol i}=\frac{\partial ^r}{\partial x_{i_1}\cdots\partial x_{i_r}},$$
so that the index $i_j$ refers to the coordinate with respect to which the $j$-th partial
derivative is performed, for $1\leq j\leq r$. As noted in Section \ref{sec:symm}, the application
$p$ gives a one-to-one correspondence between $\mathcal{PR}_{d,r}$ and $\{1,\dots,d^r\}$ so that it
induces a natural ordering ${\boldsymbol i_1}=p^{-1}(1),\dots,\boldsymbol i_{d^r}=p^{-1}(d^r)$ in
$\mathcal{PR}_{d,r}$ (this correspondence is written down explicitly in Appendix \ref{pinv}). It is
not difficult to check that, in the formal expression of $\D^{\otimes r}$ as a Kronecker power, its
coordinates are arranged precisely in that order, that is,
$$\D^{\otimes r}=(\D_{\boldsymbol i_1},\dots,\D_{\boldsymbol i_{d^r}}).$$

Alternatively, when commutation of partial derivatives is possible, the
coordinates of $\D^{\otimes r}$ can also be written as $\D^{\boldsymbol m}$ for some $\boldsymbol m\in\mathcal I_{d,r}=\big\{(m_1,\dots,m_d)\colon0\leq
m_k\leq r,\,|\boldsymbol m|=r\big\}$, where $|\boldsymbol m|=\sum_{k=1}^dm_k$ and
$$\D^{\boldsymbol m}=\frac{\partial^{|\boldsymbol m|}}{\partial x_1^{m_1}\cdots\partial x_d^{m_d}}.$$
Therefore, here the index $m_k$ refers to the number of times that it is partially differentiated with respect to
$x_k$, for $1\leq k\leq d$.

It is clear that for a given $\boldsymbol i\in\mathcal{PR}_{d,r}$ the two definitions agree if $m_k$ is set to be the number of times that the $k$-th coordinate appears in $\boldsymbol i$; that is, $m_k=\sum_{j=1}^rI_{\{i_j=k\}}$. But for a given $\boldsymbol m\in\mathcal I_{d,r}$ there might be many
possible multi-indices $\boldsymbol i\in\mathcal{PR}_{d,r}$ such that $\D^{\boldsymbol m}=\D_{\boldsymbol i}$. This is because, provided partial differentiation commutation is possible, the set $\{\D^{\boldsymbol m}\colon\boldsymbol m\in\mathcal I_{d,r}\}$ contains the unique coordinates of
$\D^{\otimes r}$.

Moreover, it is not difficult to show (for instance, by induction on $d$) that the cardinality of
$\mathcal I_{d,r}$ is $N_{d,r}=|\mathcal I_{d,r}|=\binom{r+d-1}{r}$, which is usually much smaller
than $d^r$. So an efficient way to obtain $\D^{\otimes r}$ is to compute its unique $N_{d,r}$
elements $\{\D^{\boldsymbol m}\colon\boldsymbol m\in\mathcal I_{d,r}\}$ and then rearrange them to form
$\D^{\otimes r}$.

If all the unique partial derivatives $\{\D^{\boldsymbol m}\colon\boldsymbol m\in\mathcal I_{d,r}\}$ are
collected in a vector $\mathfrak D^r$ of length $N_{d,r}$, there is also a natural ordering according to
which its coordinates should be positioned. This ordering is induced by that of $\D^{\otimes
r}=(\D_{\boldsymbol i_1},\dots,\D_{\boldsymbol i_{d^r}})$ in a way such that any $\D^{\boldsymbol m}$ can be associated
with the first value of $j\in\{1,\dots,d^r\}$ such that $\D^{\boldsymbol m}=\D_{\boldsymbol i_j}$. For instance,
the first element of $\mathfrak D^r$ is necessarily $\D^{(r,0,\dots,0)}=\D_{(1,1,\dots,1)}=\D_{\boldsymbol i_1}$
and the last element of $\mathfrak D^r$ is necessarily $\D^{(0,\dots,0,r)}=\D_{(d,d,\dots,d)}=\D_{\boldsymbol i_{d^r}}$.

For any $\boldsymbol m\in\mathcal I_{d,r}$, \citet[Section 12.8]{Erd53} showed that $\D^{\boldsymbol m}\phi_\bSigma(\bx)$ can also be expressed with the
aid of a real-valued Hermite polynomial $\mathcal H^{\boldsymbol m}$ (remember that the bold font notation $\boldsymbol{\mathcal{H}}_r$ is reserved for the vector-valued Hermite polynomial introduced in (\ref{Herm})) in a way such that
$$\D^{\boldsymbol m}\phi_\bSigma(\bx)=(-1)^{|\boldsymbol m|}\phi_\bSigma(\bx)\mathcal H^{\boldsymbol m}(\bx;\bSigma).$$

So if we denote
by $\boldsymbol{\mathfrak H}_r(\bx;\bSigma)$ the vector of length $N_{d,r}$ containing as coordinates all the values $\{\mathcal H^{\boldsymbol m}(\bx;\bSigma)\colon\boldsymbol m\in\mathcal I_{d,r}\}$, arranged in the same order as the elements of
$\mathfrak D^r\phi_\bSigma(\bx)$, it is possible to write $\mathfrak D^r\phi_\bSigma(\bx)=(-1)^r\phi_\bSigma(\bx)\boldsymbol{\mathfrak H}_r(\bx;\bSigma)$.
Thus, by comparison with the definition of $\boldsymbol{\mathcal{H}}_r$, notice that the vector $\boldsymbol{\mathfrak H}_r(\bx;\bSigma)$ contains the unique coordinates of $(\bSigma^{-1})^{\otimes r} \Herm_r(\bx; \bSigma)$.

\citet[Theorem 4.1]{Sav06} showed $d$ recursive formulas that are useful to obtain every coordinate
of $\boldsymbol{\mathfrak H}_{r+1}(\bx;\bSigma)$ from some of the elements in
$\boldsymbol{\mathfrak H}_r(\bx;\bSigma)$ and $\boldsymbol{\mathfrak H}_{r-1}(\bx;\bSigma)$.
Namely, if $\mat V=\bSigma^{-1}=(v_{ij})_{i,j=1}^d$ and $\bz=\mat V\bx=(z_1,\dots,z_d)$ then, for
$j=1,2,\dots,d$, \citet{Sav06} showed that
\begin{equation}\label{recursive}
\mathcal H^{\boldsymbol m+\boldsymbol e_j}(\bx;\bSigma)=z_j\mathcal H^{\boldsymbol m}(\bx;\bSigma)-\sum_{k=1}^dv_{jk}m_k\mathcal H^{\boldsymbol m-\boldsymbol e_k}(\bx;\bSigma),
\end{equation}
where we follow the convention that $\mathcal H^{\boldsymbol m -\boldsymbol e_k}(\bx; \bSigma) = 1$ if $m_k=0$.

Here, an algorithm is proposed to obtain recursively the whole Hermite polynomial vector of unique
elements $\boldsymbol{\mathfrak H}_{r+1}(\bx;\bSigma)$ from $\boldsymbol{\mathfrak
H}_r(\bx;\bSigma)$ and $\boldsymbol{\mathfrak H}_{r-1}(\bx;\bSigma)$, thus maintaining the analogy
with the usual univariate recursive formula. In this vector form, the recursion starts with
$\boldsymbol{\mathfrak H}_0(\bx;\bSigma)=1$ and $\boldsymbol{\mathfrak
H}_1(\bx;\bSigma)=\bSigma^{-1}\bx$.

An obvious difficulty is that the Hermite polynomial vectors of different orders have different
lengths. Besides, if all the $d$ recursive formulas are applied to each element of $\boldsymbol{\mathfrak H}_r(\bx;\bSigma)$ then $dN_{d,r}$ elements are obtained and $dN_{d,r}>N_{d,r+1}$ if $r\geq1$, so
necessarily some of the obtained elements would be duplicated. Furthermore, it would be desirable, at each step of the recursion, that the newly obtained Hermite
polynomial vector keep the correct order of its coordinates.

A recursive procedure to compute the $N_{d,r+1}$-dimensional vector $\boldsymbol{\mathfrak
H}_{r+1}(\bx;\bSigma)$ from the $N_{d,r}$-dimensional vector $\boldsymbol{\mathfrak
H}_r(\bx;\bSigma)$ and the $N_{d,r-1}$-dimensional vector $\boldsymbol{\mathfrak
H}_{r-1}(\bx;\bSigma)$, using the recursive formulas (\ref{recursive}), reads as follows:
\begin{enumerate}
\item Using (\ref{recursive}) with $j=1$ it is possible to obtain all the Hermite polynomial
    values corresponding to $\{\boldsymbol m+\boldsymbol e_1\colon\boldsymbol m\in\mathcal I_{d,r}\}=\{\mat
    m\in\mathcal I_{d,r+1}\colon m_1\geq1\}$. There are $N_{d,r}$ of them, which are put in the
    first $N_{d,r}$ positions of $\boldsymbol{\mathfrak H}_{r+1}(\bx;\bSigma)$.

    It remains to compute the Hermite values corresponding to $\{\boldsymbol m\in\mathcal
    I_{d,r+1}\colon m_1=~0\}$, which can be expressed as $\{(0,m_2,\dots,m_d)\colon
    m_2+\dots+m_d=r+1\}$. There are, therefore, $N_{d-1,r+1}$ of them, which is the remaining number of coordinates of $\boldsymbol{\mathfrak H}_{r+1}(\bx;\bSigma)$ to fill in, since $N_{d,r+1}=N_{d,r}+N_{d-1,r+1}$, according to Pascal's rule.

\item Using (\ref{recursive}) with $j=2$ it is possible to obtain all the Hermite polynomial
    values corresponding to $\{(0,m_2,\dots,m_d)\colon m_2+\dots+m_d=r+1, m_2\geq1\}$.
    Reasoning as in the first step it is clear that there are $N_{d-1,r}$ of them, which are
    obtained by adding $\boldsymbol e_2$ to the multi-indices $\boldsymbol m\in\mathcal I_{d,r}$ of the form
    $\boldsymbol m=(0,m_2,\dots,m_d)$. Since, inductively, the first $N_{d,r-1}$ coordinates of the
    vector $\boldsymbol{\mathfrak H}_r(\bx;\bSigma)$ correspond to multi-indices $\boldsymbol m\in\mathcal I_{d,r}$
    with $m_1\geq1$, formula (\ref{recursive}) with $j=2$ should be applied to the remaining
    last $N_{d,r}-N_{d,r-1}=N_{d-1,r}$ coordinates of $\boldsymbol{\mathfrak H}_r(\bx;\bSigma)$ to keep the same coherent ordering in the coordinates of $\boldsymbol{\mathfrak H}_{r+1}(\bx;\bSigma)$.

    Moreover, since formula formula (\ref{recursive}) with $j=2$ is applied to
    multi-indices $\boldsymbol m\in\mathcal I_{d,r}$ of the form $\boldsymbol m=(0,m_2,\dots,m_d)$, it can be
    further simplified to take into account that $m_1=0$, yielding
    $$\mathcal H^{\boldsymbol m+\boldsymbol e_2}(\bx;\bSigma)=z_2\mathcal H^{\boldsymbol m}(\bx;\bSigma)-\sum_{k=2}^dv_{2k}m_k\mathcal H^{\boldsymbol m-\boldsymbol e_k}(\bx;\bSigma).$$

    It remains to compute the Hermite values corresponding to $\{\boldsymbol m\in\mathcal
    I_{d,r+1}\colon m_1=m_2=0\}$, which can be expressed as $\{(0,0,m_3,\dots,m_d)\colon
    m_3+\dots+m_d=r+1\}$. There are, therefore, $N_{d-2,r+1}$ of them, which is the remaining number of coordinates of $\boldsymbol{\mathfrak H}_{r+1}(\bx;\bSigma)$ to fill in, because $N_{d,r+1}=N_{d,r}+N_{d-1,r}+N_{d-2,r+1}$, according to Pascal's rule.

\item After the $(d-1)$-th step, the first $\sum_{j=1}^{d-1}N_{d-j+1,r}$
    coordinates of $\boldsymbol{\mathfrak H}_{r+1}(\bx;\bSigma)$ have been computed, and since
    $N_{d,r+1}=1+\sum_{j=1}^{d-1}N_{d-j+1,r}$ by successive application of Pascal's rule, the
    only coordinate left is the last one, corresponding to $\boldsymbol m=(0,\dots,0,r+1)\in\mathcal
    I_{d,r+1}$. To compute it we just apply the iterative formula with $j=d$ to the last
    element of $\boldsymbol{\mathfrak H}_{r}(\bx;\bSigma)$, which corresponds to $(0,\dots,0,r)\in\mathcal I_{d,r}$,
    which in this case simplifies to
    $$\mathcal H^{(0,\dots,0,r+1)}(\bx;\bSigma)=z_d\mathcal H^{(0,\dots,0,r)}(\bx;\bSigma)-v_{dd}r\mathcal H^{(0,\dots,0,r-1)}(\bx;\bSigma).$$
\end{enumerate}

The previous steps have been merged into Algorithm \ref{unique} to derive a novel recursive procedure to compute $\D^{\otimes r}\phi_\bSigma(\bx)$.

\begin{algorithm}[h!t]
\caption{Recursive computation of $\D^{\otimes r}\phi_\bSigma(\bx)$}\label{unique}
\SetKwInOut{Input}{Input}\SetKwInOut{Output}{Output}
\Input{vector $\bx\in\mathbb R^d$, $d\times d$ matrix $\bSigma$, order $r$}
\Output{The vector $\D^{\otimes r}\phi_\bSigma(\bx)\in\mathbb R^{d^r}$}
\begin{enumerate}
\item Set $\boldsymbol{\mathfrak H}_0(\bx;\bSigma)=1$ and $\boldsymbol{\mathfrak H}_1(\bx;\bSigma)=\bSigma^{-1}\bx$
\item If $r\geq2$ then, for $k$ in $2,\dots,r$:

\qquad Proceed as in steps 1--3 in the text to obtain $\boldsymbol{\mathfrak H}_k(\bx;\bSigma)$

\qquad from $\boldsymbol{\mathfrak H}_{k-1}(\bx;\bSigma)$ and $\boldsymbol{\mathfrak H}_{k-2}(\bx;\bSigma)$

\item Distribute the elements of $\boldsymbol{\mathfrak H}_r(\bx;\bSigma)$ to form $ (\bSigma^{-1})^{\otimes r} \Herm_r(\bx;\bSigma)$
\item Return $\D^{\otimes r} \phi_\bSigma( \bx) = (-1)^r (\bSigma^{-1})^{\otimes r} \Herm_r(\bx; \bSigma) \phi_\bSigma( \bx)$
\end{enumerate}
\end{algorithm}

A natural competitor of this algorithm, also in recursive form, but based on the computation of the whole Hermite vector polynomial and not only its unique coordinates, can be derived from Theorem 7.2 in \citet{Hol96a}. From this theorem, it follows that the vectors $\boldsymbol u_k=(\bSigma^{-1})^{\otimes k}\Herm_k(\bx;\bSigma)$ satisfy the recurrence relation
\begin{equation}\label{recurrence}
\boldsymbol u_k=\S_{d,k}\big[(\bSigma^{-1}\bx)\otimes\boldsymbol u_{k-1}-(k-1)\big\{(\vec\bSigma^{-1})\otimes\boldsymbol u_{k-2}\big\}\big].
\end{equation}
Therefore, a straightforward recursive implementation of the previous formula, making use of Algorithm \ref{RecSdrv} to calculate (\ref{recurrence}), allows to obtain  $\D^{\otimes r}\phi_\bSigma(\bx)=(-1)^r\boldsymbol u_r\phi_\bSigma(\bx)$. The performance of these two recursive algorithms as well as the direct alternative is investigated in Section \ref{sec:comparisons} below.

\section{Applications to selected statistical problems}\label{sec:app}

The multivariate Gaussian density function plays a key role in many statistical problems. A number of them need not only the function itself, but some of its higher-order derivatives. In this section, we illustrate how the previous methods can be used to deal with some selected situations; the performance of the many possible algorithms arising from the application of the previous recursive techniques to each of these problems is discussed in Section \ref{sec:comparisons}.

\subsection{Moments of Gaussian random variables}\label{sec:moment}

Perhaps the most widely studied Gaussian-based scalar functions are the moments of the multivariate normal distribution and the expected values of quadratic forms in normal random variables. Many algorithms have been proposed to compute these, which are too numerous to cite all here. Surely the earliest is \citet{Iss18}, but more recent attempts include \citet{Kum73,Mag79,Gha95,Hol88,Hol96b,Tri03,Kan08} and \citet{Phi10}. As noted before, the advantage of the approach of \cite{Hol88,Hol96b} is that it produces concise explicit expressions using the symmetrizer matrix, with its corresponding computational disadvantage. The other references tend to focus on more efficient algorithmic approaches where the underlying structure is obscured, making them less amenable for
further mathematical analysis. To this end, we wish to derive algorithms which are both computationally efficient and mathematically tractable.

For $\bX \sim N_d(\bmu, \bSigma)$ a $d$-variate Gaussian random variable with mean $\bmu$ and variance $\bSigma$, its raw vector moment of order $r$ is defined as $\bmu_r=\mathbb E(\bX^{\otimes r})\in\mathbb R^{d^r}$. \cite{Hol88} showed that an explicit formula for this vector moment for an arbitrary order $r$ is given by
\begin{equation}\label{mur}
\bmu_r = r! \S_{d,r}\sum_{j=0}^{[r/2]}  \frac{1}{j! (r-2j)!2^j}
 \big\{\bmu^{\otimes (r - 2j)} \otimes  (\vec \bSigma)^{\otimes j} \big\}.
\end{equation}
Further, \citet[Equation (9.2)]{Hol96a} noted that the resemblance between the previous expression and the definition of the multivariate Hermite polynomial (\ref{Herm}) can be expressed as
\begin{equation}\label{muHerm}
\bmu_r= \Herm_r(\bmu; -\bSigma).
\end{equation}
Although the matrix $\bSigma$ in the definition of the vector Hermite polynomial needs to be positive definite so that all the formulas have a well-defined probabilistic interpretation, \citet{Hol96a} showed that (\ref{muHerm}) remains valid even if $-\bSigma$ is negative definite. Therefore, the vector moment $\bmu_r$ can be efficiently computed using the algorithms introduced in the previous sections.

Many authors, as for instance \citet{Tri03,Kan08} or \citet{Phi10}, focus instead on real-valued moments $\mu_{\boldsymbol i}=\mathbb E(X_{i_1}\cdots X_{i_r})$, where $\bX=(X_1,\dots,X_d)$ and $\boldsymbol i=(i_1,\dots,i_r)\in\mathcal{PR}_{d,r}$. The vector moment $\bmu_r$ contains all these real-valued moments as its coordinates (some of them even duplicated), but the main objection that these authors make about this vector moment formulation is about the difficulties encountered at the time of computing the symmetrizer matrix involved in (\ref{mur}), so they propose different alternatives to compute a single one of these real-valued moments. The approach described above overcomes these difficulties and allows to readily obtain all the real-valued moments at once by computing the whole vector moment.

\subsection{Quadratic forms in Gaussian random variables}
\label{sec:quad}

A closely related problem is that of computing the mixed moment of orders $(r,s)$ of two quadratic forms in normal variables, defined as
$$\nu_{r,s}(\mat A, \mat B)\equiv\nu_{r,s} (\mat A, \mat B; \bmu, \bSigma) =  \E [(\bX^\top \mat A \bX)^r(\bX^\top \mat B \bX)^s],$$
where $\mat A, \mat B$ are both $d \times d$ symmetric matrices and $\bX\sim N_d(\bmu,\bSigma)$. Note that by taking $s=0$ and $\mat B=\bI_d$ (say), the previous functional reduces to the $r$-th moment of a single quadratic form in normal variables, which will be denoted as $\nu_r(\mat A)\equiv\nu_r (\mat A; \bmu, \bSigma) =  \E [(\bX^\top \mat A \bX)^r]$.

The connection between these functionals and the vector moments of the multivariate normal distribution was highlighted by \citet{Hol96b}, who noted that
\begin{equation}\label{quadmom}
\nu_{r,s} (\mat A, \mat B)=
[(\vec^\top \mat A)^{\otimes r} \otimes (\vec^\top \mat B)^{\otimes s}] \bmu_{2r+2s}
\end{equation}
and, consequently, $\nu_{r} (\mat A)=(\vec^\top \mat A)^{\otimes r} \bmu_{2r}$. This Kronecker product
form has the advantage of decoupling the deterministic matrix product $(\vec \mat A)^{\otimes r}$ from
the raw moment $\bmu_{2r}$ of the random vector $\bX$. Moreover, Equation (\ref{quadmom}) makes it immediate to obtain a general formula for $\nu_{r,s} (\mat A, \mat B)$ for arbitrary orders $(r,s)$ from (\ref{mur}), as shown in Theorems 2 and 8 of \citet{Hol96b}, and, besides, it makes the advantage of using vector moments (as opposite to real-valued moments) more apparent. Furthermore, it also suggests a straightforward way to apply the efficient procedures for the computation of $\bmu_{2r+2s}$ in Section \ref{sec:moment} to obtain $\nu_{r,s} (\mat A, \mat B)$.

However, even if Equation (\ref{quadmom}) relates the two types of moments in a simple way, these two moments are quite different in nature. Whereas $\bmu_{2r+2s}$ is a high-dimensional vector, $\nu_{r,s} (\mat A, \mat B)$ is a scalar, so in this case it might be preferable to use an alternative recursive implementation not relying on the computation of such a high-dimensional vector.

The classical alternative approach is based on recursive relation between cumulants and lower order
$\nu$ functionals. Recall that when the cumulant generating function $\psi(t)=\log\mathbb
E[\exp\{tY\}]$ of a real random variable $Y$ is $r$-times differentiable, its $r$-th cumulant is
defined as $\psi^{(r)}(0)$ for $r\geq1$. \citet[Theorem~3.2b.2]{MP92} asserted that for $r\geq 1$,
$$
\nu_r(\mat A) = \sum_{i=0}^{r-1} \binom{r-1}{i} \kappa_{r-i}(\mat A) \nu_i(\mat A)
$$
where $\kappa_r(\mat A)$ is the $r$-th cumulant of the random variable $\bX^\top \mat A\bX$, given
by
$$\kappa_r(\mat A)\equiv\kappa_r(\mat A; \bmu, \bSigma) = 2^{r-1} (r-1)! \big [ \tr \{(\mat A \bSigma)^r\} + r \bmu^\top (\mat A \bSigma)^{r-1} \mat A \bmu
\big ].$$ The recursion starts with $\nu_0(\mat A) = 1$.

For the mixed moment $\nu_{r,s}(\mat A, \mat B) $, \citet[Equation (10)]{Smi95} showed that
\begin{equation}\label{Smith}
\nu_{r,s}(\mat A, \mat B) = \sum_{i=0}^{r} \sum_{j=0}^{s-1} \binom{r}{i} \binom{s-1}{j}
\kappa_{r-i, s-j}(\mat A, \mat B) \nu_{i,j}(\mat A, \mat B)
\end{equation}
where  $\kappa_{r,s}(\mat A,\mat B)$ is the joint $(r,s)$-th cumulant of $\bX^\top \mat A\bX$ and
$\bX^\top \mat B\bX$, which is defined as the value at $(0,0)$ of the $(r,s)$-th order partial
derivative of the joint cumulant generating function $\psi(t_1,t_2)=\log\mathbb E[\exp\{t_1\bX^\top
\mat A\bX+t_2\bX^\top \mat B\bX\}]$ for $r+s\geq1$.

\citet[Theorem~3.3.4 and Corollary 3.3.1]{MP92} provided a concise formula for $\kappa_{r,s}(\mat
A,\mat B)$ without an explicit proof. Unfortunately, although their formula is correct for some particular cases, it is wrong in general (see further details in Section \ref{sec:proof-comp}
below). We provide the correct formula for the cumulant $\kappa_{r,s}(\mat A,\mat B)$ in Theorem~\ref{jointcumulant} below.

Let us denote by $\mathcal{MP}_{r,s}$ the set of permutations of the multiset having $r$ copies of
1 and $s$ copies of 2; that is,
$$\mathcal{MP}_{r,s}=\big\{\boldsymbol i=(i_1,\dots,i_{r+s})\in\{1,2\}^{r+s}\colon n_1(\boldsymbol i)=r,\ n_2(\boldsymbol i)=s\big\},$$
where $n_\ell(\boldsymbol i)$ denotes the number of times that $\ell$ appears in $\boldsymbol i$,
for $\ell=1,2$; i.e., $n_\ell(\boldsymbol i)=\sum_{k=1}^{r+s}I_{\{i_k=\ell\}}$. Recall that the
cardinality of $\mathcal{MP}_{r,s}$ is $(r+s)!\big/(r!s!)$.

\begin{theorem}\label{jointcumulant}
For $r+s\geq1$, the joint cumulant of order $(r,s)$ of $\bX^\top \mat A\bX$ and $\bX^\top \mat B\bX$ is given by
\begin{align*}
\kappa_{r,s}(\mat A,\mat B)=2^{r+s-1}r!s!\sum_{\boldsymbol i\in\mathcal{MP}_{r,s}}\tr\big[\mat F_{i_1}\cdots\mat F_{i_{r+s}}\big\{\bI_d/(r+s)+\bSigma^{-1}\bmu\bmu^\top\big\}\big],
\end{align*}
where $\mat F_1=\mat A\bSigma$ and $\mat F_2=\mat B\bSigma$.
\end{theorem}

The combination of the correct formula for $\kappa_{r,s}(\mat A,\mat B)$ with (\ref{Smith}) results
in a straightforward recursive algorithm for the computation of $\nu_{r,s}(\mat A, \mat B)$, whose performance is reported in Section \ref{sec:comparisons}.

Upon visual inspection, these $\nu$ functionals are composed of various traces of products of $\mat A,
\mat B$ and $\bSigma$, and quadratic forms of these products in $\bmu$. Because there exist many
results for, say, the differential analysis of these scalar functions, the subsequent differential
analysis is no more difficult than the original form in terms of symmetrizer matrices. %, and yet it is also computationally efficient.

\subsection{Analysis of Gaussian kernel-based non-parametric data smoothers}
\label{sec:kde}

A general goal of non-parametric data smoothing is to generate smooth visualizations of discretized
data for exploratory data analysis, e.g. see \cite{Sim96} for an overview. Let $\bX_1, \bX_2, \dots, \bX_n$ be a random sample drawn from a common density $f$. The kernel
density estimator of $\D^{\otimes r} f$ with Gaussian kernel is $\D^{\otimes r} \hat{f}_\bH (\bx)
= n^{-1} \sum_{i=1}^n \D^{\otimes r}\phi_\bH(\bx - \bX_j)$, where $\bH$ is a positive definite
bandwidth matrix. Hence, all the techniques introduced in the previous sections are quite useful to
obtain an efficient implementation of this estimator in practice.

It should be noted that the computation of this kernel density derivative estimator can be expedited
using different and complementary approaches to ours if the bandwidth matrix $\bH$ is constrained
to being a diagonal matrix; e.g. the binned kernel estimators of \citet{Wan94}, and the fast Gauss
transform based estimators of \citet{RDZ10}. On the other hand, our goal is to produce efficient
algorithms for use with maximally general unconstrained bandwidth matrices.

The crucial factor in the performance of a kernel estimator is the selection of the bandwidth
matrix $\bH$ of smoothing parameters. The mean integrated squared error (MISE) of the kernel
density derivative estimator is defined as
$$\mathrm{MISE}_r (\bH) = \E \int \|\D^{\otimes r}\hat{f}_\bH (\bx) - \D^{\otimes r}f(\bx)\|^2 \,
d\bx,$$ where $\|\cdot\|$ denotes the usual Euclidean norm. Under suitable regularity conditions,
\citet{CDW11} showed that as $n\to\infty$ the MISE can be approximated by
$${\rm AMISE}_r(\bH)= n^{-1}|\bH|^{-1/2}\tr\big((\bH^{-1})^{\otimes r}\mathbf R(\D^{\otimes r}\phi)\big) +
\tfrac{(-1)^r}{4} [(\vec^\top\bI_{d^r}) \otimes (\vec^\top \bH)^{\otimes 2}] \bpsi_{2r+4},$$ where
$\mat R(\mat g) = \int \mat g(\bx) \mat g(\bx)^\top \ d\bx$ for a vector-valued function $\mat g$,
and $\bpsi_{s} = \int \D^{\otimes s} f(\bx) f(\bx) \, d\bx$. Thus the minimizer of ${\rm AMISE}_r$
is a bandwidth matrix with an asymptotically optimal $L_2$ risk.

The usual approach to select the bandwidth matrix $\bH$ from the data is based on first estimating
the MISE using the data sample, and then selecting the bandwidth that minimizes the obtained
estimate of the MISE. Here, the step regarding the estimation of the MISE typically involves the
computation of $V$-statistics of degree 2 based on higher order derivatives of the Gaussian density
function. For instance, the three methods for bandwidth selection proposed in \citet{CD12} are
based, respectively, on the following three estimators of the MISE
\begin{align*}
\CV_r(\bH)&=(-1)^r\vec^\top\bI_{d^r}\bigg\{n^{-2}\sum_{i,j=1}^n\D^{\otimes2r}\phi_{2\bH}(\bX_i-\bX_j)-2[n(n-1)]^{-1}\sum_{i\neq j}\D^{\otimes 2r}\phi_\bH(\bX_i-\bX_j)\bigg\}\\
\PI_r(\bH)&= n^{-1}|\bH|^{-1/2}\tr\big\{(\bH^{-1})^{\otimes r}\mathbf R(\D^{\otimes r}\phi)\big\} +
\tfrac{(-1)^r}{4} [(\vec^\top\bI_{d})^{\otimes r} \otimes (\vec^\top \bH)^{\otimes 2}] \hat \bpsi_{2r+4}(\bG)\\
\SCV_r(\bH)&=n^{-1}|\bH|^{-1/2}\tr\big\{(\bH^{-1})^{\otimes r}\mathbf R(\D^{\otimes r}\phi)\big\} \\
&\quad+(-1)^r\vec^\top\bI_{d^r}n^{-2}\sum_{i,j=1}^n\D^{\otimes2r}\big\{\phi_{2\bH+2\bG}-2\phi_{\bH+2\bG}+\phi_{2\bG}\big\}(\bX_i-\bX_j)
\end{align*}
where $\hat{\bpsi}_{s} (\bG) = n^{-2} \sum_{i,j=1}^n \D^{\otimes s} \phi_\bG (\bX_i - \bX_j)$ is a
kernel estimator of $\bpsi_s$ for a given even number $s$, based on a pilot bandwidth matrix $\bG$. These estimators of the MISE
are commonly referred to as cross validation, plug-in and smoothed cross validation criteria,
respectively.

The zero-th order derivative case poses little problem for computation. However, if higher order
derivatives are considered, we quickly run into computational difficulties. \citet{Lin10} reduced
the computational burden of general $U$-statistics by aggregating $U$-statistics of random
sub-samples. Here, a different approach is taken by seeking computationally efficient forms for the
full sample, restricted to $V$-statistics of degree 2 based on derivatives of the Gaussian density
function.

Let us denote $\eta_{r,s}(\bx; \mat B, \bSigma) = [(\vec^\top \bI_d)^{\otimes r} \otimes (\vec^\top
\mat B)^{\otimes s}] \allowbreak \D^{\otimes 2r+2s} \phi_{\bSigma} (\bx)$ for a $d \times d$
symmetric matrix $\mat B$. Define also $\eta_{r}(\bx; \bSigma) \equiv \eta_{r,0}(\bx; \bI_d,
\bSigma)  = (\vec^\top \bI_d)^{\otimes r} \D^{\otimes 2r} \phi_{\bSigma} (\bx)$. It is easy to show
that the previous bandwidth selection criteria can be expressed using these functions, so that
\begin{align*}
\CV_r(\bH)&=(-1)^r \bigg\{n^{-2}\sum_{i,j=1}^n \eta_{r} (\bX_i-\bX_j; 2\bH)
-2[n(n-1)]^{-1}\sum_{i\neq j} \eta_{r} (\bX_i-\bX_j; \bH) \bigg\},\\
\mathrm{PI}_r(\bH) &=2^{-(d+r)} \pi^{-d/2}  n^{-1}|\bH|^{-1/2} \nu_r(\bH^{-1}; \boldsymbol{0},\bI_d)
+ \tfrac{(-1)^r}{4} n^{-2}\sum_{i,j=1}^n \eta_{r, 2}(\bX_i - \bX_j; \bH, \bG),\\
\mathrm{SCV}_r(\bH) &= 2^{-(d+r)} \pi^{-d/2}n^{-1}|\bH|^{-1/2}  \nu_r(\bH^{-1};\boldsymbol{0}, \bI_d)
+ (-1)^r n^{-2} \sum_{i,j=1}^n \big\{\eta_{r}(\bX_i - \bX_j; 2\bH+2\bG) \\
&\quad -2\eta_{r}(\bX_i - \bX_j; \bH+2\bG) + \eta_{r}(\bX_i - \bX_j; 2\bG)\big\},
\end{align*}
where it should be noted that the equivalence in the first term of the plug-in and smoothed cross
validation criteria follows from Lemma 3.$iv)$ in \citet{CDW11}.

Thus, the key for an efficient implementation of these criterion is to develop a fast recursive algorithm to compute the $\eta$ functionals. All the developments in the previous sections can be used for this goal by taking into account the following new result.

\begin{theorem}
\label{thm:eta} For a fixed $\bx$, the previous $\eta$ functionals are related to the $\nu$ functionals as follows
\begin{align*}
\eta_{r} (\bx; \bSigma)
&= \phi_{\bSigma} (\bx) \nu_r\big(\bI_d; \bSigma^{-1} \bx , -\bSigma^{-1}\big) \\
\eta_{r,s} (\bx; \mat B, \bSigma) &=
\phi_{\bSigma} (\bx) \nu_{r,s}\big(\bI_d, \mat B; \bSigma^{-1} \bx , -\bSigma^{-1}\big).
\end{align*}
\end{theorem}

The recursive formulation allows for a more efficient optimization algorithm to obtain the minimizer of the
corresponding bandwidth selection criteria, and these minimizers are commonly used as the basis for
data-based optimal bandwidth matrices, whose asymptotic and finite sample properties were studied
in \citet{CD12}.

For large $n$, evaluating the double sum in the previous $V$-statistics can pose two different, in
some sense dual, problems. If we enumerate singly the data difference $\bX_i - \bX_j$, then this
increases the computation time in $n^2$. If we wish to take advantage of vectorized computations
offered in many software packages, then this requires storing an $n^2 \times d$ matrix in memory
which is not always feasible. Thus we have to find the right compromise between execution speed and
memory usage on commonly available desktops computers.

Following Theorem~\ref{thm:eta}, any $V$-statistic of the form $Q_r
(\bSigma)=n^{-2}\sum_{i,j=1}^n\eta_r(\bX_i-\bX_j;\bSigma)$ can be decomposed as a double sum of
products of $\nu_r\big(\bI_d;  \bSigma^{-1} (\bX_i - \bX_j), -\bSigma^{-1}\big)$ with
$\phi_{\bSigma} (\bX_i - \bX_j)$. The two most computationally intensive steps involve the
cumulants
\begin{align*}
\kappa_r\big(\bI_d; \bSigma^{-1} (\bX_i - \bX_j), -\bSigma^{-1}\big)
&= (-2)^{r-1} (r-1)! \big \lbrace -\tr (\bSigma^{-r}) + (\bX_i - \bX_j)^\top \bSigma^{-r-1} (\bX_i - \bX_j)\big \rbrace
\end{align*}
and the normal densities
\begin{align*}
\phi_\bSigma(\bX_i - \bX_j) &= (2\pi)^{-d/2} |\bSigma|^{-1/2} \exp \big\{-\tfrac{1}{2} (\bX_i - \bX_j)^\top \bSigma^{-1} (\bX_i - \bX_j)\big\}.
\end{align*}
The time consuming step in common is the double sum of the terms of the form $(\bX_i - \bX_j)^\top
\bSigma^{-\ell} (\bX_i - \bX_j)$ for some power $\ell\geq1$ of $\bSigma^{-1}$. If we decouple this term into
components
\begin{equation} \label{decouple}
(\bX_i -\bX_j)^\top \bSigma^{-\ell} (\bX_i - \bX_j)
 = \bX_i^\top \bSigma^{-\ell} \bX_i + \bX_j^\top \bSigma^{-\ell} \bX_j - 2 \bX_i^\top \bSigma^{-\ell} \bX_j,
\end{equation}
then each of them are efficiently handled by software in terms of execution but with memory requirements only
slightly larger than storing the original sample $\bX_1, \dots ,\bX_n$, since the differences $\bX_i - \bX_j$, $j=1, \dots, n$, are
kept in memory for each $i$ singly rather than all for $i$ as we loop over $i$.

\section{Numerical comparisons}\label{sec:comparisons}

The implementation of all the algorithms described in this paper are contained in the \texttt{ks} library \citep{ks} in the \texttt{R} statistical programming language \citep{R13}, and in a separate, specific script (Online Resource~1) which is also available from the authors' websites. For each scenario, each algorithm was executed 10 times in {\tt R} 3.0.1  under Ubuntu 12.04 LTS 64 bits, installed on a Dell Precision T6700 with 8 Intel Xeon E5-2609 @ 2.40 GHz CPUs and 32 Gb RAM.
Since the actual execution times are highly dependent on the computing set-up used, it is more useful to focus on relative
execution times to indicate likely performance gains on other computing set-ups.

\subsection{Symmetrizer matrix}

A carefully designed algorithm for the direct implementation was used so that, in fact, only one of the two loops in (\ref{expdef2}) is needed, which moreover selects to loop over $i=1,\dots,d^r$ if $d^r<r!$ (with $r!$ the cardinality of $\mathcal P_r$), and over $\sigma\in\mathcal P_r$ otherwise.
This direct implementation based on Equation~(\ref{expdef2}) was compared to the recursive implementation in Algorithm~\ref{RecSym},
where the ratio of mean direct execution time to the
mean recursive execution time are presented in Table~\ref{tab:Sym} for dimension $d=2,3,4$ and
order $r=2,4,6,8$. Due to memory restrictions, the symmetrizer matrix for $d=4, r=8$ was not
able to be computed.

\begin{table}[!htp]
\centering
\begin{tabular}{llrrrr}
& & $r=2$ & $r=4$ & $r=6$ & $r=8$ \\
$d=2$ %& direct (sec) & 0.0006 & 0.0015 & 0.0313 & 6.5689 \\
& direct/recursive   & 0.42 & 0.75 & 9.78 & 386.49 \\
$d=3$ %& direct (sec) & 0.0003 & 0.0022 & 0.1730 & 68.7226 \\
& direct/recursive   & 0.20 & 0.71 & 0.66 & 0.52   \\
$d=4$ %& direct (sec) & 0.0005 & 0.0041 & 1.1645 &  -- \\
& direct/recursive   & 0.52 & 0.36 & 0.04 &  --
\end{tabular}
\caption{Comparison of mean execution times for direct and recursive implementations
to compute the symmetrizer matrix $\S_{d,r}$, for dimension $d=2,3,4$ and
order $r=2,4,6,8$. Each row is the ratio of the mean direct time to the
mean recursive time.}
\label{tab:Sym}
\end{table}

%As expected, for a fixed value of $d$ the recursive implementation eventually becomes
%faster once $r$ is large enough. For instance,%
For $d=2$, the recursive algorithm seems
to be faster from $r=6$ on, and is already more than 300 times faster for $r=8$. Thus,
for low values of $d$ and large $r$, the recursive implementation is preferable.
But as $d$ increases it is harder to notice the advantage of the recursive approach,
since it is noticeable only for large values of $r$. Certainly, it must be pointed out that
the direct computation of the simplified form (\ref{expdef2}) makes it quite competitive for
low values of $d$, which are the most commonly used in practice, since handling these
huge matrices with the current computational power seems inadvisable for $d\geq5$ unless
$r$ is very low.

In fact, using the direct formula (\ref{expdef2}) can also be useful to alleviate the problem of
the storage space needed, because these sparse matrices have a tiny proportion of non-zero
elements, especially for higher values of $d$. Since the symmetrizer matrices are symmetric
\citep{Sch03}, specifying only its lower triangular part (including the diagonal) suffices to
recover the whole matrix. Figure \ref{fig:1} displays the proportion of the $d^r(d^r+1)/2$ entries
in the lower triangular part of $\S_{d,r}$ that are not null. Thus, for instance, only 70 elements
need to be stored to recover $\S_{7,2}$ (which has 2401 entries), and only 9801 elements are needed
to recover $\S_{6,4}$ (which has $1\,679\,616$ entries). It remains as an interesting open problem
to find an explicit formula for the number of non-zero entries of $\S_{d,r}$.

\begin{figure}[t]
\centering

\begin{tikzpicture}
    \node[anchor=south west,inner sep=0] at (0,0) {\includegraphics[width=0.5\textwidth]{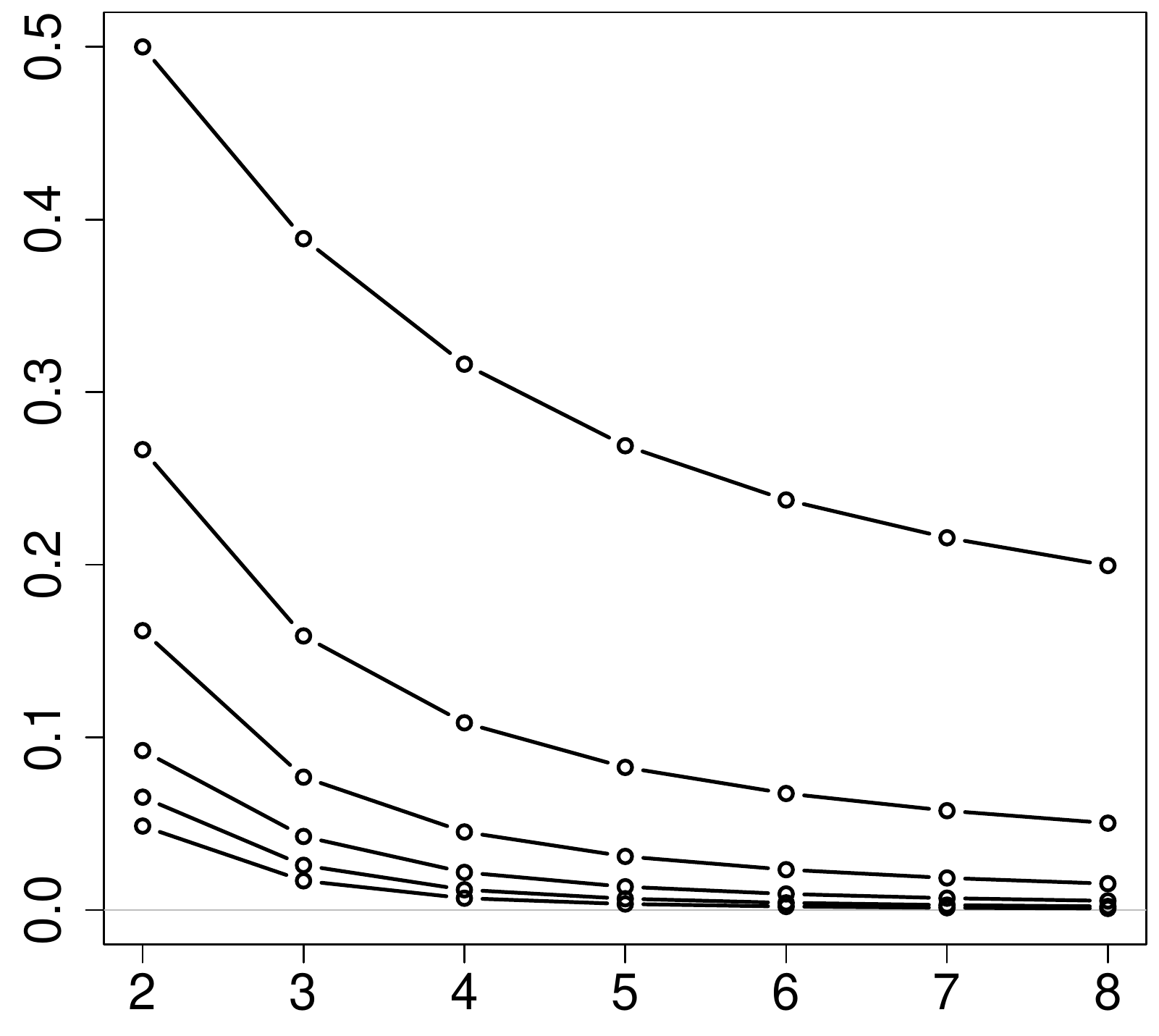}};
    \draw[color=black] (4.1,-.1) node {$r$};
    \draw[color=black] (-0.4,6) node[left,rotate=90] {Prop. of non-zero elements};
\end{tikzpicture}

\caption{Proportion of non-zero elements in the lower triangular part of $\S_{d,r}$ as a function of $r$. From top to bottom the lines correspond to $d=2,\dots,7.$ }
\label{fig:1}
\end{figure}

\subsection{Product of a symmetrizer matrix and a vector}
A similar experiment was conducted to compare the computation
times of the direct approach to compute the product of a symmetrizer matrix
and a vector, which is based on Equation (\ref{Sdrv}), and the recursive
approach presented in Algorithm \ref{RecSdrv}. From Table~\ref{tab:Sdrv}, for this problem, the recursive approach
proved to be faster than the direct one in all the scenarios. Moreover, the reductions in
time achieved by the recursive algorithm can be extremely large for values of $d$ and $r$
commonly encountered in practice. For instance, for $r=8$ the recursive algorithm
produced a result in about 1000--2000 times faster. In the previous section, the
symmetrizer matrix $\S_{4,8}$ was not able to be computed, whereas the product $\S_{4,8} \boldsymbol{v}$,
for $\boldsymbol{v}=(1,2, \dots, 4^8)$ posed no memory problems.

\begin{table}[!ht]
\centering
\begin{tabular}{llrrrr}
& & $r=2$ & $r=4$ & $r=6$ & $r=8$ \\
$d=2$ & direct/recursive & 4.00 & 2.83 & 22.85 & 1878.65 \\
$d=3$ & direct/recursive & 3.50 & 2.00 & 28.03 & 2595.11 \\
$d=4$ & direct/recursive & 2.00 & 2.20 & 21.35 & 1150.85
\end{tabular}
\caption{Comparison of mean execution times for direct and recursive implementations
to compute the product of the symmetrizer matrix $\S_{d,r}$ and a $d^r$-vector, for dimension
$d=2,3,4$ and
order $r=2,4,6,8$. Each row is the ratio of the mean direct time to the
mean recursive time.}
\label{tab:Sdrv}
\end{table}

\subsection{Derivatives of a Gaussian density function}
We compared the performance of computing  the $r$-th derivative of $d$-variate standard Gaussian density
$\D^{\otimes r} \phi(1,\ldots,1)$, for dimension $d=2,3,4$ and order $r=2,4,6,8,10$.
In Table~\ref{tab:deriv}, the upper row in each pair of rows compares
the direct implementation based on Equations~(\ref{Drphi}) and (\ref{Herm}), which nevertheless makes use of Algorithm \ref{RecSdrv} to obtain the multiplication by the symmetrizer matrix,
to the first recursive algorithm based on Equation~(\ref{recurrence}).
The lower row in each pair of rows compares
the direct implementation to the second recursive algorithm based on Algorithm~\ref{unique}
where only the unique elements are computed.
We observed that computing the unique elements of the Hermite vector
polynomial eventually becomes faster, as $r$ increases, than the direct and/or first recursive implementations.

\begin{table}[!htp]
\centering

\begin{tabular}{llrrrrr}
& & $r=2$ & $r=4$ & $r=6$ & $r=8$ & $r=10$\\
$d=2$
& direct/recursive  & 1.18 & 0.68 & 0.66 & 0.59 & 0.66\\
& direct/unique     & 1.93 & 0.76 & 0.83 & 1.12 & 8.00\\
$d=3$
& direct/recursive  & 0.92 & 0.77 & 0.65 & 0.93 & 0.79\\
& direct/unique     & 0.58 & 0.48 & 0.62 & 3.01 & 7.46\\
$d=4$
& direct/recursive  & 1.00 & 0.81 & 0.94 & 0.96 & 0.85\\
& direct/unique     & 0.48 & 0.32 & 0.83 & 3.93 & 7.96
\end{tabular}
\caption{Comparison of mean execution times for direct and recursive implementations
to compute $\D^{\otimes r} \phi(\cdot)$ the $r$-th derivative of $d$-variate standard Gaussian density,
for $d=2,3,4$ and order $r=2,4,6,8,10$. In each pair of rows, the upper row
is the ratio of mean direct time to the
mean time of the first recursive implementation, and the lower row is the ratio of the
mean direct time to the mean time of the second recursive implementation where
only the unique elements are computed.}
\label{tab:deriv}
\end{table}

\subsection{Moments of a Gaussian random variable}

We compared the performance of computing the vector $r$-th moment $\bmu_r$ for a
standard normal Gaussian $N(0, \bI_d)$, for dimension $d=2,3,4$ and order $r=2,4,6,8,10$.
From Table~\ref{tab:mur}, there does not appear to be a clearly more efficient
implementation. In the upper row in each pair of rows, comparing
the direct implementation to the first recursive algorithm based on Equations~(\ref{recurrence})
and (\ref{muHerm}), many of these time ratios are around
one, indicating a more or less equal computational load. In the lower row in each pair of rows,
comparing the direct implementation to the second recursive algorithm where
only the unique elements are computed based on Algorithm~\ref{unique} and Equation~(\ref{muHerm}),
the latter tends to be more efficient than the direct and the first
recursive implementation.

\begin{table}[!ht]
\centering
\begin{tabular}{llrrrrr}
& & $r=2$ & $r=4$ & $r=6$ & $r=8$ & $r=10$\\
$d=2$
& direct/recursive  & 1.79 & 0.99 & 0.80 & 0.82 & 0.96 \\
& direct/unique     & 4.67 & 3.69 & 1.43 & 1.20 & 2.63 \\
$d=3$
& direct/recursive  & 3.23 & 1.37 & 1.10 & 1.25 & 1.13 \\
& direct/unique     & 6.45 & 1.00 & 0.66 & 2.69 & 7.14 \\
$d=4$
& direct/recursive  & 4.31 & 1.05 & 0.76 & 1.06 & 1.02 \\
& direct/unique     & 3.11 & 0.79 & 0.57 & 4.13 & 7.78
\end{tabular}
\caption{Comparison of mean execution times for direct and recursive implementations
to compute  $\bmu_r$ the vector moment of a $d$-variate standard Gaussian random variable,
for $d=2,3,4$ and order $r=2,4,6,8,10$. In each pair of rows, the upper row
is the ratio of mean direct time to the
mean time of the first recursive implementation, and the lower row is the ratio of the
mean direct time to the mean time of the second recursive implementation where
only the unique elements are computed.}
\label{tab:mur}
\end{table}

\subsection{Expected value of quadratic forms in Gaussian random variables}
Recall that the expected value of $(r,s)$-the product of the quadratic form for a
$d$-variate Gaussian random variable
$\bX$ is
$\nu_{r,s} (\mat A, \mat B) =  \E [(\bX^\top \mat A \bX)^r(\bX^\top \mat B \bX)^s]$,
and that $\nu_{r,s}$ involves the $(2r+2s)$-th moments of $\bX$,
so we investigated the performance for dimension
$d=2,3,4$ and $(r,s)$ such that $1 \leq s \leq r, r+s\leq 5$, with
$\mat A = \operatorname{diag}(1,\dots, d), \mat B = \operatorname{diag}(d,\dots, 1)$.
In the upper row in each group of rows of Table~\ref{tab:nurs}, comparing
the direct implementation based on Equation~(\ref{quadmom}) to the first recursive algorithm
based on Equations~(\ref{recurrence}), (\ref{muHerm}--\ref{quadmom}),
most of these time ratios are around
one, indicating a more or less equal computational load.
In the middle row in each group of rows,
comparing the direct implementation to the second recursive algorithm based
on Algorithm~\ref{unique} and Equations~(\ref{muHerm}--\ref{quadmom}), where
only the unique elements of the vector moment are computed,
most of these time ratios are around one for $r+s < 4$, and greater than one for $r+s \geq4$.
In the lower row in each group of rows,
comparing the direct implementation to the third recursive algorithm
based on the moment-cumulant results of Equation~(\ref{Smith})
and Theorem~\ref{jointcumulant},
the computational speed was multiplied by 10- to 1000-fold for many cases, as
$d$ and/or $(r+s)$ increase.  For the $(r,s)$ pairs considered,
this third cumulants-based recursive form is generally the most efficient approach, with more and more substantial speed-ups as the dimension and/or the derivative order increase.

\begin{table}[!ht]
\centering
\begin{tabular}{llrrrrrr}
& & \multicolumn{6}{c}{$(r,s)$} \\
& & $(1,1)$ & $(2,1)$ & $(2,2)$ & $(3,1)$ & $(3,2)$ & $(4,1)$ \\
$d=2$
& direct/recursive  & 1.36 & 1.22 & 1.20 & 0.90 & 1.08  &1.08 \\
& direct/unique     & 0.75 & 0.87 & 1.18 & 0.87 & 2.44 & 2.37 \\
& direct/cumulant   & 1.07 & 1.47 & 1.64 & 1.83 & 3.51 & 5.33 \\
$d=3$
& direct/recursive  & 0.94 & 1.06 &  1.25 &  1.25 &   1.14 &   1.14 \\
& direct/unique     & 0.40 & 0.56 &  2.86 &  2.84 &   7.76 &   7.74 \\
& direct/cumulant   & 1.33 & 2.70 & 18.21 & 24.68 & 228.98 & 336.00 \\
$d=4$
& direct/recursive  & 1.11 &  1.05 &   1.35 &   1.10 &    0.89&     0.88 \\
& direct/unique     & 0.27 &  0.74 &   4.93 &   4.63 &    7.17&     7.15 \\
& direct/cumulant   & 1.33 & 10.20 & 290.50 & 334.77 & 3797.22 & 4823.30
\end{tabular}
\caption{Comparison of mean execution times for direct and recursive implementations
to compute $\nu_{r,s}$ the expected value of $(r,s)$-th product of the quadratic form
of a $d$-variate Gaussian random variable,
for $d=2,3,4$ and orders $(r,s), 1 \leq s \leq r, r+s\leq 5$.
In each group of rows, the upper row
is the ratio of mean direct time to the
mean time of the first recursive implementation, the middle row is the ratio of
the mean direct time to the mean time of the second recursive implementation where
only the unique elements of the vector moment are computed, and the last row is the ratio of
the mean direct time to the mean time of the third recursive implementation based on
moment-cumulants.}
\label{tab:nurs}
\end{table}

\subsection{Gaussian kernel based $V$-statistics}
Samples of size $n=100,1000,10000$ were drawn from the $d$-variate standard
Gaussian distribution $N(\mat 0, \bI_d)$, for $d=2,3,4$, and from these samples
the $V$-statistic $Q_{r} (\bI_d)
= n^{-2} \sum_{i,j=1}^n \eta_{r} (\bX_i - \bX_j; \bI_d)$ was computed for $r=0,2,4$.
The direct implementation is based on Equations~(\ref{Drphi}) and (\ref{Herm})
and the recursive cumlants-based algorithm combining Theorem~{\ref{thm:eta} and
Equation~(\ref{decouple}). As expected, Table \ref{tab:Qr} shows that the time
savings increase with increasing dimension and increasing derivative order, with
10- to 1000-fold improvements in most cases.

\begin{table}[t]
\centering
\begin{tabular}{llrrrrrrrrr}
& & \multicolumn{3}{c}{$n=100$} & \multicolumn{3}{c}{$n=1000$} & \multicolumn{3}{c}{$n=10000$} \\
& & $r=0$ & $r=2$ & $r=4$ & $r=0$ & $r=2$ & $r=4$ & $r=0$ & $r=2$ & $r=4$  \\
$d=2$ & direct/cumulant & 0.55 & 2.93 &   8.48 & 6.86 & 23.97 &  44.33 & 4.85  & 17.47 &  59.33\\
$d=3$ & direct/cumulant & 0.99 & 4.04 & 118.05 & 7.87 & 42.51 & 163.42 & 7.07  & 32.17 & 258.38\\
$d=4$ & direct/cumulant & 1.57 & 7.06 & 347.71 & 9.05 & 64.52 & 548.14 & 6.36  & 46.83 & 1010.99
\end{tabular}
\caption{Comparison of mean execution times for direct and recursive implementations
to compute $Q_r$ the Gaussian kernel based $V$-statistic, for dimension
$d=2,3,4$, derivative order $r=0,2,4$, and sample size $n=100,1000,10000$. Each row is the ratio of the mean direct time to the
mean recursive time based on moment-cumulants.}
\label{tab:Qr}
\end{table}

\bigskip

\noindent {\bf Acknowledgements.} We thank two anonymous referees for a careful reading of the paper. This work has been partially supported by grants MTM2010-16660 (both authors) and MTM2010-17366 (first author) from the Spanish Ministerio de Ciencia e Innovaci\'on. The second author also received funding from the program ``Investissements d’avenir'' ANR-10-IAIHU-06.

\appendix

\section{Appendix: Proofs}\label{proofs}

\subsection{Proofs of the results in Section \ref{sec:symm}}

\allowdisplaybreaks

The key elements to prove Theorem \ref{thm:1} are the following two lemmas.

\begin{lemma}\label{perm}
For every $j\in\mathbb N_{r+1}:=\{1,2,\dots,r+1\}$ denote by $\tau_j\in\mathcal P_{r+1}$ the
permutation defined by $\tau_j(j)=r+1$, $\tau_j(r+1)=j$ and $\tau_j(i)=i$ for $j\neq i\neq r+1$.
Then we can express
$$\mathcal P_{r+1}=\big\{\sigma\circ\tau_j\colon \sigma\in\mathcal P_r,\
j\in\mathbb N_{r+1}\big\}.$$
\end{lemma}

\begin{proof}
As any $\sigma\in\mathcal P_r$ can be thought as an element of $\mathcal P_{r+1}$ by defining
$\sigma(r+1)=r+1$, consider the map $\varphi\colon\mathcal P_r\times\mathbb N_{r+1}\to\mathcal
P_{r+1}$ given by $\varphi(\sigma,j)=\sigma\circ\tau_j$. We conclude by noting that this map is
bijective, with inverse given by $\varphi^{-1}(\tilde \sigma)=(\sigma,j)$, where
$j=\tilde\sigma^{-1}(r+1)$ is such that $\tilde\sigma(j)=r+1$ and, for $i\in\mathbb N_r$,
$\sigma(i)=\tilde\sigma(i)$ if $\tilde\sigma(i)\neq r+1$ and $\sigma(i)=\tilde\sigma(r+1)$ if
$\tilde\sigma(i)=r+1$.
\end{proof}

\begin{lemma}\label{trans}
If $\mat A\in\mathcal M_{m\times n}$, $\mat B\in\mathcal M_{p\times q}$ and
$\gvec{a},\gvec{b}\in\mathbb R^d$, then
$$\mat A\otimes\gvec{a}^\top\otimes\mat B\otimes\gvec{b}^\top=(\mat A\otimes\gvec{b}^\top\otimes\mat B\otimes\gvec{a}^\top)\cdot(\bI_{dn}\otimes\mat K_{q,d})
\cdot(\bI_n\otimes\mat K_{d,dq}).$$
\end{lemma}

\begin{proof}
Use the properties of the commutation matrix to first permute $\gvec{a}^\top\otimes\mat B$ with
$\gvec{b}^\top$, keeping $\mat A$ in the same place, and then to permute $\gvec{a}^\top$ with $\mat B$
keeping $\mat A\otimes\gvec{b}^\top$ in the same place.
\end{proof}

The previous lemmas are helpful to manipulate the original definition of $\S_{d,r}$ and thus obtain the proof of Theorem \ref{thm:1}.

\begin{proof}[Proof of Theorem \ref{thm:1}]
Note that for any two vectors $\boldsymbol v,\boldsymbol w\in\mathbb R^d$ we have $\boldsymbol
v\boldsymbol w^\top=\boldsymbol v\otimes\boldsymbol w^\top$. Then, with the identification $\mathcal
P_r\subset\mathcal P_{r+1}$ and the notation $\tau_j$ as in Lemma \ref{perm}, for any
$\sigma\in\mathcal P_r$ and $j\in\mathbb N_{r+1}$,
\begin{align}
\bigotimes_{\ell=1}^{r+1}\be_{i_\ell}\be_{i_{\sigma(\tau_j(\ell))}}^\top&=
\bigotimes_{\ell=1}^{j-1}\be_{i_\ell}\be_{i_{\sigma(\ell)}}^\top\otimes
\be_{i_j}\be_{i_{r+1}}^\top\otimes
\bigotimes_{\ell=j+1}^{r}\be_{i_\ell}\be_{i_{\sigma(\ell)}}^\top\otimes
\be_{i_{r+1}}\be_{i_{\sigma(j)}}^\top\nonumber\\
&=\Big\{\bigotimes_{\ell=1}^r\be_{i_\ell}\be_{i_{\sigma(\ell)}}^\top\otimes
\be_{i_{r+1}}\be_{i_{r+1}}^\top\Big\}\cdot(\bI_{d^j}\otimes\mat
K_{d^{r-j},d})(\bI_{d^{j-1}}\otimes\mat K_{d,d^{r-j+1}})\label{eq1}
\end{align}
where for the second equality we have applied Lemma \ref{trans} with $\boldsymbol a=\be_{i_{r+1}}$,
$\boldsymbol b=\be_{i_{\sigma(j)}}$,
$$\mat
A=\bigotimes_{\ell=1}^{j-1}\be_{i_\ell}\be_{i_{\sigma(\ell)}}^\top\otimes
\be_{i_j}\in\mathcal M_{d^j,d^{j-1}}\quad\text{and}\quad\mat
B=\bigotimes_{\ell=j+1}^{r}\be_{i_\ell}\be_{i_{\sigma(\ell)}}^\top\otimes
\be_{i_{r+1}}\in\mathcal M_{d^{r-j+1},d^{r-j}}.$$ Taking Lemma
\ref{perm}, (\ref{eq1}) and the definition of $\mat T_{d,r+1}$ into account,
\begin{align*}
\S_{d,r+1}&=\frac{1}{(r+1)!}\sum_{i_1,i_2,\dots,i_{r+1}=1}^d\; \sum_{\sigma\in\mathcal
P_{r+1}}\bigotimes_{\ell=1}^{r+1}\be_{i_\ell}\be_{i_{\sigma(\ell)}}^\top\\
&=\frac{1}{(r+1)!}\sum_{i_1,i_2,\dots,i_{r+1}=1}^d\;
\sum_{\sigma\in\mathcal
P_r}\sum_{j=1}^{r+1}\bigotimes_{\ell=1}^{r+1}\be_{i_\ell}\be_{i_{\sigma(\tau_j(\ell))}}^\top\\
&=\frac{1}{r!}\sum_{i_1,i_2,\dots,i_{r+1}=1}^d\;
\sum_{\sigma\in\mathcal
P_r}\Big\{\bigotimes_{\ell=1}^r\be_{i_\ell}\be_{i_{\sigma(\ell)}}^\top\otimes
\be_{i_{r+1}}\be_{i_{r+1}}^\top\Big\}\mat T_{d,r+1}\\
&=\Big\{\S_{d,r}\otimes
\Big(\textstyle\sum_{i_{r+1}=1}^d\be_{i_{r+1}}\be_{i_{r+1}}^\top\Big)\Big\}\mat
T_{d,r+1}\\
&=(\S_{d,r}\otimes\bI_d) \mat T_{d,r+1},
\end{align*}
as  $\bI_d=\sum_{i=1}^d\be_i\be_i^\top$.
\end{proof}

To obtain a recursive formula for the matrix $\mat T_{d,r}$  we first need to write the matrices
$\mat K_{d^{p+1},d}$ and $\mat K_{d,d^{p+1}}$ depending on $\mat K_{d^p,d}$ and $\mat K_{d,d^p}$,
respectively.

\begin{lemma}\label{Kmat}
For any $p\geq0$
\begin{align*}
\mat K_{d^{p+1},d}&=(\bI_{d^p}\otimes\mat K_{d,d})(\mat
K_{d^p,d}\otimes\bI_d)=(\bI_d\otimes\mat K_{d^p,d})(\mat
K_{d,d}\otimes\bI_{d^p})\\
\mat K_{d,d^{p+1}}&=(\mat
K_{d,d^p}\otimes\bI_d)(\bI_{d^p}\otimes\mat K_{d,d})=(\mat
K_{d,d}\otimes\bI_{d^p})(\bI_d\otimes\mat K_{d,d^p}).
\end{align*}
\end{lemma}

\begin{proof}
Using part $i)$ of Theorem 3.1 in \cite{MN79}, we can write
\begin{align*}
\mat
K_{d^{p+1},d}&=\sum_{j=1}^d(\be_j^\top\otimes\bI_{d^{p+1}}\otimes\be_j)=\sum_{j=1}^d(\be_j^\top\otimes\bI_{d^p}\otimes\bI_d\otimes\be_j)\\
&=(\bI_{d^p}\otimes\mat
K_{d,d})\sum_{j=1}^d(\be_j^\top\otimes\bI_{d^p}\otimes\be_j\otimes\bI_d)=(\bI_{d^p}\otimes\mat
K_{d,d})(\mat K_{d^p,d}\otimes\bI_d).
\end{align*}
The second equality for $\mat K_{d^{p+1},d}$ follows similarly and the formulas for $\mat
K_{d,d^{p+1}}$ can be derived from the previous ones by noting that $\mat K_{d,d^{p+1}}=\mat
K_{d^{p+1},d}^\top$.
\end{proof}

Using the previous lemma we obtain a straightforward proof of Theorem \ref{thm:2}.

\begin{proof}[Proof of Theorem \ref{thm:2}]
Using Lemma \ref{Kmat} for the first $r-1$ terms in the definition of $\mat T_{d,r+1}$, and the property that $(\mat A\mat C)\otimes(\mat B\mat D)=(\mat A\otimes\mat B)(\mat C\otimes\mat D)$, it follows that
\begin{align*}
(r+1)\mat T_{d,r+1}&=\sum_{j=1}^{r-1}(\bI_{d^j}\otimes\mat
K_{d^{r-j},d})(\bI_{d^{j-1}}\otimes\mat K_{d,d^{r-j+1}})+(\bI_{d^{r-1}}\otimes\mat K_{d,d})+\bI_{d^{r+1}}\\
&=\sum_{j=1}^{r-1}\big[\bI_{d^j}\otimes\{(\bI_{d^{r-j-1}}\otimes\mat K_{d,d})(\mat
K_{d^{r-j-1},d}\otimes\bI_d)\}\big]\\&\qquad\times\big[\bI_{d^{j-1}}\otimes\{(\mat K_{d,d^{r-j}}\otimes\bI_d)(\bI_{d^{r-j}}\otimes\mat K_{d,d})\}\big]+(\bI_{d^{r-1}}\otimes\mat K_{d,d})+\bI_{d^{r+1}}\\
&=(\bI_{d^{r-1}}\otimes\mat K_{d,d})\Big[\Big\{\sum_{j=1}^{r}(\bI_{d^j}\otimes\mat
K_{d^{r-j-1},d})(\bI_{d^{j-1}}\otimes\mat K_{d,d^{r-j}})\Big\}\otimes\bI_d\Big]\\&\qquad\times(\bI_{d^{r-1}}\otimes\mat K_{d,d})+(\bI_{d^{r-1}}\otimes\mat K_{d,d})\\
&=(\bI_{d^{r-1}}\otimes\mat K_{d,d})(r\mat T_{d,r}\otimes\bI_d)(\bI_{d^{r-1}}\otimes\mat K_{d,d})+(\bI_{d^{r-1}}\otimes\mat K_{d,d}),
\end{align*}
where the third equality makes use of $\bI_{d^p}\otimes\bI_{d^q}=\bI_{d^{p+q}}$.
\end{proof}

\subsection{Proofs of the results in Section \ref{sec:4}}

As noted in the text, the proof of Corollary \ref{cor:1} follows by induction on $r$.

\begin{proof}[Proof of Corollary \ref{cor:1}]
For $r=1$ the formula immediately follows, since $\S_{d,1}=\bI_d=\mat T_{d,1}$. The induction step is easily deduced by using formula $\S_{d,r+1}=(\S_{d,r}\otimes\bI_d)\mat T_{d,r+1}$ from Theorem \ref{thm:1} using the same tools as before, taking into account that $\bI_{d^p}\otimes\bI_d=\bI_{d^{p+1}}$ and that $(\mat A\mat C)\otimes(\mat B\mat D)=(\mat A\otimes\mat B)(\mat C\otimes\mat D)$.
\end{proof}

Corollary \ref{cor:2} is deduced from Corollary \ref{cor:1} as follows.

\begin{proof}[Proof of Corollary \ref{cor:2}]
Clearly, the Kronecker product $\bigotimes_{\ell=1}^r\be_{i_\ell}$ of $r$ vectors $\be_{i_1},\dots,\be_{i_r}$ of the canonical base of $\mathbb R^d$ gives the $p(i_1,\dots,i_r)$-th vector of the canonical base in $\mathbb R^{d^r}$ (i.e., the $p(i_1,\dots,i_r)$-th column of $\bI_{d^r}$). Therefore, any vector $\boldsymbol v=(v_1,\dots,v_{d^r})\in\mathbb R^{d^r}$ can be written as $\boldsymbol v=\sum_{i=1}^{d^r}v_i\bigotimes_{\ell=1}^r\be_{(p^{-1}(i))_\ell}$ and so, by linearity, it suffices to obtain a simple formula for expressions of the type $(\mat T_{d,k}\otimes\bI_{d^{r-k}})(\bigotimes_{\ell=1}^r\be_{i_\ell})$. Further, since $(\mat T_{d,k}\otimes\bI_{d^{r-k}})(\bigotimes_{\ell=1}^r\be_{i_\ell})=\big\{\mat T_{d,k}\big(\bigotimes_{\ell=1}^k\be_{i_\ell}\big)\big\}\otimes\bigotimes_{\ell=k+1}^r\be_{i_\ell}$, it follows that it is enough to provide a simple interpretation for the multiplications $\mat T_{d,k}\big(\bigotimes_{\ell=1}^k\be_{i_\ell}\big)$ for $k=2,\dots,r$.

Finally, using the properties of the commutation matrix \citep{MN79}, it can be checked that
\begin{equation}\label{Tdke}
\mat T_{d,k}\bigg(\bigotimes_{\ell=1}^k\be_{i_\ell}\bigg)=\frac1k\sum_{j=1}^k\bigg\{\bigotimes_{\ell=1}^{j-1}\be_{i_\ell}\otimes\be_{i_k}\otimes
\bigotimes_{\ell=j+1}^{k-1}\be_{i_\ell}\otimes\be_{i_j}\bigg\}
\end{equation}
with the convention that $\bigotimes_{\ell=j}^k\be_{i_\ell}=1$ if $j>k$. In words, $k\mat T_{d,k}\big(\bigotimes_{\ell=1}^k\be_{i_\ell}\big)$ consists of adding up all the possible $k$-fold Kronecker products in which the last factor is interchanged with the $j$-th factor, for $j=1,2,\dots,k$.
\end{proof}

\subsection{Proofs of the results in Section \ref{sec:app}}\label{sec:proof-comp}

First, let us point out why the formula for the joint cumulant in Corollary 3.3.1 of \cite{MP92} is not always correct. Using the notation of Theorem \ref{jointcumulant} above, their formula reads as follows: for $r\geq1,\,s\geq1$,
\begin{align}
\kappa_{r,s}(\mat A, \mat B) &= 2^{r+s-1} (r+s-1)!\tr \big(\mat F_1^r\mat F_2^s\big)\nonumber\\
&\quad+2^{r+s-1}(r+s-2)!\big\{ r(r-1) \tr\big(\mat F_1^{r-1}\mat F_2^s\mat F_1\bSigma^{-1}\bmu\bmu^\top \big)\label{MP}\\
&\quad +s(s-1) \tr\big(\mat F_2^{s-1}\mat F_1^r\mat F_2\bSigma^{-1}\bmu\bmu^\top \big)+2rs \tr\big(\mat F_1^r\mat F_2^s\bSigma^{-1}\bmu\bmu^\top \big) \big\}.\nonumber
\end{align}
To further simplify our comparison, consider for example the case $\bmu=0$, and $r=s=2$, so that (\ref{MP}) simply reads $ 2^3\, 6\tr \big(\mat F_1^2\mat F_2^2\big)$. Writing down explicitly the six elements in $\mathcal{MP}_{2,2}$ and applying the cyclic property of the trace, the correct form from Theorem \ref{jointcumulant} has
\begin{align*}
2^3\, 2!2!\sum_{\boldsymbol i\in\mathcal{MP}_{2,2}}\tr\big(\mat F_{i_1}\mat F_{i_2}\mat F_{i_3}\mat F_{i_4}\big)/4
=2^3\ \big\{4\tr\big(\mat F_1^2\mat F_2^2\big)+2\tr\big(\mat F_1\mat F_2\mat F_1\mat F_2\big)\big\}%\label{MP3}
\end{align*}
instead. Both formulas involve $6$ traces of matrices, all having two factors $\mat F_1$ and another two factors $\mat F_2$. However, despite the aforementioned cyclic property of the trace, it is not true in general that $\tr\big(\mat F_1\mat F_2\mat F_1\mat F_2\big)=\tr\big(\mat F_1^2\mat F_2^2\big)$, and that causes an error in formula (\ref{MP}). A similar argument shows the reason why some of the terms involving $\bmu$ in (\ref{MP}) are also wrong.

A sufficient condition for formula (\ref{MP}) to be correct is that $\mat F_1\mat F_2=\mat F_2\mat
F_1$. If that condition holds, then the correct formula for the joint cumulant further simplifies
to
$$
\kappa_{r,s}(\mat A,\mat B)=2^{r+s-1}(r+s-1)!\big\{\tr\big(\mat F_1^r\mat F_2^s\big)+(r+s)\tr\big(\mat F_1^r\mat F_2^s\bSigma^{-1}\bmu\bmu^\top\big)\big\}.
$$

The proof of Theorem \ref{jointcumulant} is based on Matrix Calculus. Let us introduce some further notation to simplify the calculations. For $i=1,2$, denote
$$\mat C_i\equiv\mat C_i(t_1,t_2)=(\bI_d-2t_1\mat F_1-2t_2\mat F_2)^{-1}\mat F_i$$ and, similarly, $\mat C_3\equiv\mat C_3(t_1,t_2)=(\bI_d-2t_1\mat F_1-2t_2\mat F_2)^{-1}\bSigma^{-1}\bmu\bmu^\top$. Taking into account the formula for the differential of the inverse of a matrix given in \citet[][Chapter 8]{MN99}, notice that the introduced notation allows for simple expressions for the following differentials: for any $i\in\{1,2,3\}$ and $j\in\{1,2\}$, $d\mat C_i=2\mat C_j\mat C_idt_j.$ In words, differentiating any of these matrix functions with respect to $t_j$ consists on pre-multiplying by $2\mat C_j$.

More generally, for $i_1,\dots,i_r\in\{1,2\}$, $j\in\{1,2\}$ and $m\in\{1,2,3\}$ we have
\begin{align}
d(\mat C_{i_1}\cdots\mat C_{i_r}\mat C_m)&=\{d(\mat C_{i_1}\cdots\mat C_{i_r})\}\mat C_m+\mat C_{i_1}\cdots\mat C_{i_r}d\mat C_m\nonumber\\
&=2\Big\{\sum_{\ell=1}^r\Big(\prod_{k=1}^{\ell-1}\mat C_{i_k}\Big)(\mat C_j\mat C_{i_\ell})\Big(\prod_{k=\ell+1}^{r}\mat C_{i_k}\Big)\mat C_m
+\mat C_{i_1}\cdots\mat C_{i_r}\mat C_j\mat C_m\Big\}dt_j\nonumber\\
&=2\sum_{\ell=1}^{r+1}\Big(\prod_{k=1}^{\ell-1}\mat C_{i_k}\Big)\mat C_j\Big(\prod_{k=\ell}^{r}\mat C_{i_k}\Big)\mat C_m\, dt_j,\label{dprod}
\end{align}
where $\prod_{k=a}^b\mat C_{i_k}$ is to be understood as $\bI_d$ if $a>b$.

The key tool for the proof of Theorem \ref{jointcumulant} is the following lemma, which is indeed valid for any matrix function having the properties of $\mat C_m$ exhibited above.

\begin{lemma}\label{difw}
For any $m\in\{1,2,3\}$, consider the function $w(t_1,t_2)=\tr \mat C_m$. Then,
\begin{equation*}%\label{conj}
\frac{\partial^{r+s}}{\partial t_1^{r}\partial t_2^{s}}w(t_1,t_2)=2^{r+s}\,r!s!\sum_{\boldsymbol i\in\mathcal{MP}_{r,s}}\tr\big(\mat C_{i_1}\cdots\mat C_{i_{r+s}}\mat C_m\big).
\end{equation*}
\end{lemma}

\begin{proof}
From (\ref{dprod}) it easily follows that $d^r\mat C_m=2^rr!\,\mat C_1^r\mat C_m\, dt_1^r$, so that
\begin{equation*}%\label{start}
\frac{\partial^r}{\partial t_1^r}w(t_1,t_2)=2^r\,r!\,\tr\big(\mat C_1^r\mat C_m\big).
\end{equation*}
Hence, to conclude what we need to show is that, for $s=0,1,2,\ldots$,
\begin{equation}\label{induc}
\frac{\partial^{s}}{\partial t_2^{s}}\tr\big(\mat C_1^{r}\mat C_m\big)=2^{s}\,s!\sum_{\boldsymbol i\in\mathcal{MP}_{r,s}}\tr\big(\mat C_{i_1}\cdots\mat C_{i_{r+s}}\mat C_m\big)
\end{equation}
To prove (\ref{induc}) we proceed by induction on $s$, since the initial step corresponding to $s=0$ is clear. Assuming that (\ref{induc}) is true for the $(s-1)$-th derivative, the induction step consists of showing that the formula also holds for the $s$-th derivative; that is,
\begin{equation}\label{induc2}
\sum_{\boldsymbol i\in\mathcal{MP}_{r,s-1}}\frac{\partial}{\partial t_2}\tr\big(\mat C_{i_1}\cdots\mat C_{i_{r+s-1}}\mat C_m\big)=2s\sum_{\boldsymbol i\in\mathcal{MP}_{r,s}}\tr\big(\mat C_{i_1}\cdots\mat C_{i_{r+s}}\mat C_m\big).
\end{equation}
Taking into account (\ref{dprod}), to prove (\ref{induc2}) it suffices to show that the set
\begin{align*}\label{Aset}
\mathcal A_{r,s}&=\bigcup_{\ell=1}^{r+s}\{(i_1,\dots,i_{\ell-1},2,i_{\ell},\dots,i_{r+s-1})\colon \boldsymbol i\in\mathcal{MP}_{r,s-1}\}\\
&=\{(2,i_1,\dots,i_{r+s-1})\colon \boldsymbol i\in\mathcal{MP}_{r,s-1}\}\cup\{(i_1,2,\dots,i_{r+s-1})\colon \boldsymbol i\in\mathcal{MP}_{r,s-1}\}\\&\quad\cup\cdots\cup\{(i_1,\dots,i_{r+s-1},2)\colon \boldsymbol i\in\mathcal{MP}_{r,s-1}\}
\end{align*}
coincides precisely with the multiset that contains $s$ copies of each of the elements of $\mathcal{MP}_{r,s}$. This can be showed as follows: it is clear that all the elements in $\mathcal A_{r,s}$ belong to $\mathcal{MP}_{r,s}$. On the hand, notice that any vector $\boldsymbol i=(i_1,\dots,i_{r+s})\in\mathcal{MP}_{r,s}$ contains the number 2 in exactly $s$ of its coordinates, which can be distributed along any of the $r+s$ positions. If one of those number 2 coordinates is deleted from $\boldsymbol i$, the resulting vector belongs to $\mathcal{MP}_{r,s-1}$, and repeating that process for all the $s$ coordinates with the number 2, then $s$ copies of $\boldsymbol i$ are found $\mathcal A_{r,s}$.
\end{proof}

Making use of Lemma \ref{difw} next we prove Theorem~\ref{jointcumulant}.

\begin{proof}[Proof of Theorem~\ref{jointcumulant}]
\cite{Mag86} showed that the joint cumulant generating function of $\bX^\top \mat A\bX$ and $\bX^\top \mat B\bX$ can be written as
$\psi(t_1,t_2)=u(t_1,t_2)-\tfrac12\bmu^\top\bSigma^{-1}\bmu+v(t_1,t_2)$, where
\begin{align*}
u(t_1,t_2)&=-\tfrac12\log |\bI_d-2t_1\mat F_1-2t_2\mat F_2|\qquad\text{and}\\
v(t_1,t_2)&=\tfrac12\tr\big\{(\bI_d-2t_1\mat F_1-2t_2\mat F_2)^{-1}\bSigma^{-1}\bmu\bmu^\top\big\},
\end{align*}
with $\mat F_1=\mat A\bSigma$ and $\mat F_2=\mat B\bSigma$. Since for $r+s\geq1$ the $(r,s)$-th joint cumulant is defined as $\kappa_{r,s}(\mat A,\mat B)=\frac{\partial^{r+s}}{\partial t_1^{r}\partial t_2^s}\psi(0,0)$, it suffices to show that
$$\frac{\partial^{r+s}}{\partial t_1^{r}\partial t_2^s}\psi(t_1,t_2)=2^{r+s-1}r!s!\sum_{\boldsymbol i\in\mathcal{MP}_{r,s}}\tr\big[\mat C_{i_1}\cdots\mat C_{i_{r+s}}\big\{\bI_d/(r+s)+\mat C_3\big\}\big].$$

With the previous notations, $v(t_1,t_2)=\frac12\tr\mat C_3$, so Lemma \ref{difw} immediately yields the desired formula for the second summand.

For the first one, combining the chain rule with the formula for the differential of a determinant given in \citet[][Chapter 8]{MN99}, it follows that $\frac{\partial}{\partial t_1}u(t_1,t_2)=\tr \mat C_1$. So, applying Lemma \ref{difw} to $\frac{\partial}{\partial t_1}u(t_1,t_2)$, we obtain
$$\frac{\partial^{r+s}}{\partial t_1^{r}\partial t_2^{s}}u(t_1,t_2)=2^{r+s-1}\,(r-1)!s!\sum_{\boldsymbol i\in\mathcal{MP}_{r-1,s}}\tr\big(\mat C_{i_1}\cdots\mat C_{i_{r+s-1}}\mat C_1\big).$$
By the symmetry in $(t_1,t_2)$ and $(r,s)$ of the preceding argument we come to
\begin{multline*}
(r+s)\times\frac{\partial^{r+s}}{\partial t_1^{r}\partial t_2^{s}}u(t_1,t_2)\\=2^{r+s-1}\,r!s!\Big\{\sum_{\boldsymbol i\in\mathcal{MP}_{r-1,s}}\tr\big(\mat C_{i_1}\cdots\mat C_{i_{r+s-1}}\mat C_1\big)+\sum_{\boldsymbol i\in\mathcal{MP}_{r,s-1}}\tr\big(\mat C_{i_1}\cdots\mat C_{i_{r+s-1}}\mat C_2\big)\Big\}.
\end{multline*}
The proof is finished by noting that, clearly,
\begin{equation*}
\mathcal{MP}_{r,s}=\{(i_1,\dots,i_{r+s-1},1)\colon\boldsymbol i\in\mathcal{MP}_{r-1,s}\}\cup\{(i_1,\dots,i_{r+s-1},2)\colon\boldsymbol i\in\mathcal{MP}_{r,s-1}\}.
\qedhere
\end{equation*}
\end{proof}

Although Theorem \ref{thm:eta} suffices to obtain a fast recursive implementation of the CV, PI and
SCV criteria, here a slightly more general version of this result is shown. Let us denote
$\widetilde\eta_{r,s}(\bx; \mat A, \mat B, \bSigma) = [(\vec^\top \mat A)^{\otimes r} \otimes
(\vec^\top \mat B)^{\otimes s}] \allowbreak \D^{\otimes 2r+2s} \phi_{\bSigma} (\bx)$ for $d \times
d$ symmetric matrices $\mat A, \mat B$ and also $\widetilde\eta_{r}(\bx; \mat A, \bSigma) \equiv
\widetilde\eta_{r,0}(\bx; \mat A, \bI_d, \bSigma)  = (\vec^\top \mat A)^{\otimes r} \D^{\otimes 2r}
\phi_{\bSigma} (\bx)$. Notice that the $\eta$ functionals can be seen to be particular cases of the
$\widetilde\eta$ functionals by setting $\mat A=\bI_d$.

\begin{theorem}
\label{thm:eta2} For a fixed $\bx$, the previous $\widetilde\eta$ functionals are related to the
$\nu$ functionals as follows
\begin{align*}
\widetilde\eta_{r} (\bx; \mat A, \bSigma)
&= \phi_{\bSigma} (\bx) \nu_r\big(\mat A; \bSigma^{-1} \bx , -\bSigma^{-1}\big) \\
\widetilde\eta_{r,s} (\bx; \mat A, \mat B, \bSigma) &=
\phi_{\bSigma} (\bx) \nu_{r,s}\big(\mat A, \mat B; \bSigma^{-1} \bx , -\bSigma^{-1}\big).
\end{align*}
\end{theorem}

\begin{proof}
Notice that (\ref{muHerm}) entails $\nu_r(\mat A;\bmu,\bSigma)=(\vec^\top \mat A)^{\otimes
r}\boldsymbol{\mathcal H}_{2r}(\bmu;-\bSigma)$. And from Theorem 3.1 in \cite{Hol96a},
$(\bSigma^{-1})^{\otimes 2r}\boldsymbol{\mathcal H}_{2r}(\bx;\bSigma)=\boldsymbol{\mathcal
H}_{2r}(\bSigma^{-1}\bx;\bSigma^{-1})$. Therefore,
\begin{align*}
\widetilde\eta_r(\bx;\mat A,\bSigma)&=(\vec^\top \mat A)^{\otimes r} \D^{\otimes 2r} \phi_{\bSigma}(\bx) =\phi_\bSigma(\bx)(\vec^\top \mat A)^{\otimes r}(\bSigma^{-1})^{\otimes 2r}\boldsymbol{\mathcal
H}_{2r}(\bx;\bSigma)\\
&=\phi_\bSigma(\bx)(\vec^\top \mat A)^{\otimes r}\boldsymbol{\mathcal
H}_{2r}(\bSigma^{-1}\bx;\bSigma^{-1})=\phi_\bSigma(\bx)\nu_r\big(\mat A; \bSigma^{-1} \bx , -\bSigma^{-1}\big),
\end{align*}
as desired. The proof for $\widetilde\eta_{r,s}$ follows analogously.
\end{proof}

\section{Appendix: Generation of all the permutations with repetitions}\label{pinv}

A preliminary step to the methods described in Sections \ref{sec:symm}, \ref{sec:4} and
\ref{sec:unique} involves generating the set of all the permutations with repetitions
$\mathcal{PR}_{d,r}$. This set can be portrayed as a matrix $\mat P$ of order $d^r\times r$, whose
$(i,j)$-th entry represents the $j$-th coordinate of the $i$-th permutation in
$\mathcal{PR}_{d,r}$.

Moreover, in view of Section \ref{sec:unique} it seems convenient to keep the natural order of
these permutations induced by the formulation $\mathcal{PR}_{d,r}=\{p^{-1}(i)\colon
i=1,\dots,d^r\}$. Hence, in our construction the vector $p^{-1}(i)$ will constitute the $i$-th row
of $\mat P$.

Let $\lfloor x\rfloor$ denote the integer part of a real number $x$, that is, the largest integer
not greater than $x$. Then, if $i=p(i_1,\dots,i_r)=1+\sum_{j=1}^r(i_j-1)d^{j-1}$ with
$i_1,\dots,i_r\in\{1,\dots,d\}$, it is not hard to show that
$\lfloor(i-1)/d^{k-1}\rfloor=\sum_{j=k}^r(i_j-1)d^{j-k}$ for $k=1,\dots,r$, so that the $j$-th
coordinate of the vector $\boldsymbol i=(i_1,\dots,i_r)=p^{-1}(i)$ can be expressed as $i_j=\lfloor
(i-1)/d^{j-1}\rfloor-d\lfloor (i-1)/d^j\rfloor+1$.

Thus, assuming there is a {\tt floor}$()$ function available that can be applied in an element-wise
form to a matrix and returns the integer part of each of its entries, the set $\mathcal{PR}_{d,r}$
is efficiently obtained as the matrix $\mat P={\tt floor}(\mat Q_{-(r+1)})-d\cdot{\tt floor}(\mat
Q_{-1})+1$, where $\mat Q$ is a $d^r\times (r+1)$ matrix whose $(i,j)$-th entry is $(i-1)/d^{j-1}$
for $i=1,\dots,d^r$ and $j=1,\dots,r+1$, and $\mat Q_{-k}$ refers to the sub-matrix obtained from
$\mat Q$ by deleting the $k$-th column.

\bibliographystyle{apalike}

\end{document}